\documentclass[journal,comsoc]{IEEEtran}
\usepackage[algo2e]{algorithm2e}
\usepackage[super]{nth}
\usepackage{cite}
\usepackage{graphicx}  
\usepackage{amsmath}  
\usepackage{algorithm}
\usepackage[latin1]{inputenc}
\usepackage[T1]{fontenc}
\usepackage{times}
\usepackage{tabularx}
\usepackage{stfloats}
\usepackage{fancybox}
\usepackage{verbatim}
\usepackage{longtable}
\usepackage{url}
\usepackage{boxedminipage}
\usepackage{amssymb}
\usepackage{tabu}
\usepackage{pdfpages}
\usepackage{epstopdf}
\usepackage{cite}
\usepackage{algpseudocode}
\usepackage{optidef}
\usepackage{lscape}
\usepackage{xcolor}  
\usepackage{subcaption}
\usepackage[font=small,labelfont=bf]{caption}
\usepackage{mathtools}
\usepackage{float}

\newtheorem{lemma}{Lemma}

\newcommand{\be}{\begin{equation}}
\newcommand{\ee}{\end{equation}}
\newcommand{\bal}{\begin{aligned}}
\newcommand{\eal}{\end{aligned}}	
\newcommand{\ba}{\begin{array}}
\newcommand{\ea}{\end{array}}
\newcommand{\bea}{\begin{eqnarray}}
\newcommand{\eea}{\end{eqnarray}}

\newcommand{\vbar}{\raisebox{.17ex}{\rule{.04em}{1.35ex}}}
\newcommand{\vbarind}{\raisebox{.01ex}{\rule{.04em}{1.1ex}}}
\newcommand{\R}{\ifmmode {\rm I}\hspace{-.2em}{\rm R} \else ${\rm I}\hspace{-.2em}{\rm R}$ \fi}
\newcommand{\T}{\ifmmode {\rm I}\hspace{-.2em}{\rm T} \else ${\rm I}\hspace{-.2em}{\rm T}$ \fi}
\newcommand{\N}{\ifmmode {\rm I}\hspace{-.2em}{\rm N} \else \mbox{${\rm I}\hspace{-.2em}{\rm N}$} \fi}
\newcommand{\B}{\ifmmode {\rm I}\hspace{-.2em}{\rm B} \else \mbox{${\rm I}\hspace{-.2em}{\rm B}$} \fi}
\newcommand{\Hil}{\ifmmode {\rm I}\hspace{-.2em}{\rm H} \else \mbox{${\rm I}\hspace{-.2em}{\rm H}$} \fi}
\newcommand{\C}{\ifmmode \hspace{.2em}\vbar\hspace{-.31em}{\rm C} \else \mbox{$\hspace{.2em}\vbar\hspace{-.31em}{\rm C}$} \fi}
\newcommand{\Cind}{\ifmmode \hspace{.2em}\vbarind\hspace{-.25em}{\rm C} \else \mbox{$\hspace{.2em}\vbarind\hspace{-.25em}{\rm C}$} \fi}
\newcommand{\Q}{\ifmmode \hspace{.2em}\vbar\hspace{-.31em}{\rm Q} \else \mbox{$\hspace{.2em}\vbar\hspace{-.31em}{\rm Q}$} \fi}
\newcommand{\Z}{\ifmmode {\rm Z}\hspace{-.28em}{\rm Z} \else ${\rm Z}\hspace{-.28em}{\rm Z}$ \fi}

\newcommand{\var}{\mbox {var}}

\newcommand{\bgamma}{\mbox {\boldmath $\gamma$}}

\newcommand{\CN}{\mathcal{CN}}

\DeclarePairedDelimiter{\ceil}{\lceil}{\rceil}
\newtheorem{proof}{Proof}
\IEEEoverridecommandlockouts

\begin{document}

\title{Optimized Power Control for Massive MIMO with Underlaid D2D Communications}

\author{Amin Ghazanfari, Emil Bj{\"o}rnson, and Erik G. Larsson\\
	Department of Electrical Engineering (ISY), Link{\"o}ping University, Sweden\\
	Email:\{amin.ghazanfari, emil.bjornson, erik.g.larsson\}@liu.se
	\thanks{A preliminary version of this work was presented at the 22nd international ITG workshop on smart antennas (WSA 2018) \cite{amin2018wsa}. This paper was supported by the European Union's Horizon 2020 research
		and innovation programme under grant agreement No 641985 (5Gwireless).}}

\maketitle
\begin{abstract}
In this paper, we consider device-to-device (D2D) communication that is underlaid in a multi-cell massive multiple-input multiple-output (MIMO) system and propose a new framework for power control and pilot allocation. In this scheme, the cellular users (CUs) in each cell get orthogonal pilots which are reused with reuse factor one across cells, while all the D2D pairs share another set of orthogonal pilots. We derive a closed-form capacity lower bound for the CUs with different receive processing schemes. In addition, we derive a capacity lower bound for the D2D receivers and a closed-form approximation of it. We provide power control algorithms to maximize the minimum spectral efficiency (SE) and maximize the product of the signal-to-interference-plus-noise ratios in the network. Different from prior works, in our proposed power control schemes, we consider joint pilot and data transmission optimization. Finally, we provide a numerical evaluation where we compare our proposed power control schemes with the maximum transmit power case and the case of conventional multi-cell massive MIMO without D2D communication. Based on the provided results, we conclude that our proposed scheme increases the sum SE of multi-cell massive MIMO networks.
\end{abstract}

\section{Introduction}
\label{introduction}
Device-to-device (D2D) underlay communication and massive multiple-input multiple-output (MIMO) are two new promising technologies in wireless communication that will appear in 5G networks \cite{WWB5G}. Although these two technologies show significant improvements in the energy and spectral efficiency of the network, their combination can even enhance the network performance further. D2D underlay communication enhances the spectrum utilization by reusing the cellular resources for direct communication between D2D pairs when the transmitter and receiver are closely located. It provides benefits such as cellular traffic offloading for cellular networks. This gain is due to the direct communication between the transmitter and receiver of D2D pairs instead of sending the data through cellular base stations (BS); higher data rate and lower transmission power between D2D users due to the short-range communication \cite{asadi2014survey,doppler2009device,lin2014overview}. Hence, D2D underlay communication increases the spectral and energy efficiency of the cellular networks. These benefits come at the cost of causing extra interference to the cellular users (CUs) \cite{janis2009interference,AminICC,fodor2012design}.

Massive MIMO is one of the most significant technologies in 5G, as it offers large improvements in spectral and energy efficiency of cellular networks. Because it utilizes a large number of antennas at each BS, it offers multiplexing gains and spatial interference suppression for CUs \cite{redbook,bjornson2017massive,larsson2014massive}. By mitigating the extra interference that the D2D communication causes to the cellular network, we can potentially enhance the mentioned benefits of D2D underlay communication to the cellular network. Massive MIMO seems to address the shortcoming of D2D underlay communication. Hence, the combination of these two technologies has received considerable attention. The idea of combining D2D underlay and massive MIMO technologies has been investigated in both single-cell and multi-cell setups with various objectives, such as data power optimization, proposing new pilot allocation scheme, etc. Although, downlink (DL) D2D underlay massive MIMO has been considered in the literature \cite{shalmashi2016energy,amin2015power}, most of these works focus on the uplink (UL) D2D underlay setup. This is due to the fact that interference suppression is easier to implement in the UL at the massive MIMO BS.

Most prior works in D2D underlaying UL of massive MIMO systems has been investigated mainly for single-cell setups \cite{liu2017pilot, xu2016power, xu2017pilot,hau2017pilot,Xu2017}. In these papers, the transmitter and receiver of a D2D pair are always physically located in the same cell. This simplifies the resource allocation since each D2D pair shares resources with that specific cell. These papers consider problems such as pilot reuse schemes for D2D underlay \cite{liu2017pilot,xu2017pilot,hau2017pilot,Xu2017,xu2016power}, data power control for D2D users only \cite{hau2017pilot,xu2017pilot,xu2016power}, and data power control for both D2D and cellular users \cite{Xu2017}. The multi-cell scenario has been studied in a few papers such as \cite{lin2015interplay, He2017SE}.


Lin \textit{et al.} \cite{lin2015interplay} investigate the interplay between massive MIMO and underlaid D2D for UL data transmission in a multi-cell massive MIMO setup where the number of D2D pairs in each cell follows a Poisson distribution. A detailed analysis is provided for the case of perfect channel state information (CSI), in which they study the asymptotic and non-asymptotic spectral efficiency (SE) of CUs and D2D pairs. They also perform channel estimation based on an orthogonal training scheme that allocates orthogonal pilots to CUs and a limited number of D2D pairs nearest to the location of the BS in their cell. In this case, they study the asymptotic and non-asymptotic SE of CUs only. It is assumed that all users transmit with a predefined transmit power. In \cite{He2017SE}, the authors consider a multi-cell massive MIMO UL D2D underlay model and consider the spatial location of BSs and D2D transmitters follow two independent homogeneous Poisson point process. They derive exact expressions of the SEs for cellular and D2D communication which are not in closed form. Then they apply open-loop power control for interference cancellation in the network. This paper does not consider channel estimation and pilot transmission and they only apply a simple open-loop power control scheme.

\subsection{Contributions of the paper}
We consider a multi-cell massive MIMO setup, in contrast to most prior works, with D2D underlaying in the UL. It is assumed that the CUs in a cell have orthogonal pilots that are reused with reuse factor one between cells. In addition, the system contains multiple D2D pairs that are arbitrarily located, do not belong to any cell, and share a network-wide set of orthogonal pilots. Power control is mandatory to achieve good performance in such networks, but this has not been considered in prior works. The existing SE expressions for D2D communication depend on the small-scale fading coefficients and are unsuitable for power control since it should exploit the time-frequency diversity against fading instead of counteracting the instantaneous fading realizations. To overcome these  issues, we make the following main contributions:
\begin{itemize}
	\item  We derive a new lower bound on the ergodic capacity of the D2D receivers under imperfect CSI. The bound contains an expectation with respect to the small-scale fading.
	To enable the derivation of tractable power control schemes, we derive a closed-form approximation of the capacity bound and interpret its structure.
	\item We derive closed-form capacity lower bounds for the CUs with either maximum-ratio (MR) and zero-forcing (ZF) processing, which support arbitrary pilot allocation among the D2Ds.
	\item We formulate several data power control problems using our new SE expressions (lower bounds on the capacity) and solve them using convex optimization. We jointly optimize the data power of CUs and D2D pairs to guarantee max-min fairness to the cellular and D2D communications. We also maximize the product of the signal-to-interference-plus-noise ratios (SINRs) of the cellular and D2D users, which gives priority to users with good channels.
\end{itemize}
For the first time in this area, we also consider the joint optimization of the data transmission power and pilot transmission power:
\begin{itemize}	
	\item We formulate novel joint data and pilot power control problems with MR processing at the BSs and solve both max-min fairness and max product SINR problems using geometric programming.
	\item We propose a successive approximation algorithm to solve the corresponding joint data and pilot power control problem for the case of ZF processing at the BSs.
\end{itemize}

We provide a numerical performance evaluation of our proposed algorithms. The conference version of this work is found in \cite{amin2018wsa}, where max-min fairness power control is considered with variable data transmission power only. This paper contains more general and complete results.
\begin{figure}[h!]
	\centering
	\centerline{\includegraphics[width=7cm]{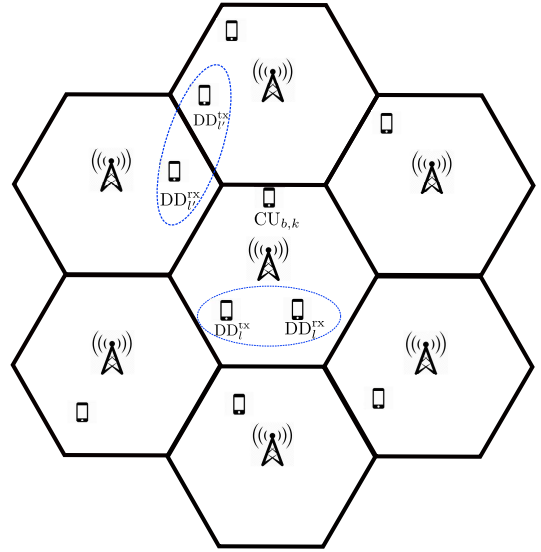}}
	\caption{The considered setup with a cellular network Massive MIMO system that is underlaid by D2D communications.}
	\label{fig:model}
\end{figure}
\section{System Model}
\label{SystemModel}
In this paper, we consider a multi-cell massive MIMO system consisting of $B$ BSs, each equipped with an array of $M$ antennas and each BS serves $K$ single-antenna CUs. The system also contains $L$ D2D pairs that are spread over the whole network and do not belong to any specific cell.\footnote{The discovery of D2D pairs is typically done at higher layers than the physical/MAC layer and is thus not considered in this paper. However, when a new D2D pair is formed or a previous pair leaves the D2D mode, the system can be optimized again using the methods described in this paper.} The network model is illustrated in Fig.~\ref{fig:model} with hexagonal cells. Note that this figure is only for illustrative purpose, while our analysis applies to cellular networks with arbitrary shapes of the cells. In this setup, we investigate D2D communication that is underlaid the UL data transmission of the multi-cell massive MIMO system. The wireless channels are varying over time and frequency, which we model by conventional block fading. We define the coherence interval of a channel as the time-frequency block in which the channel is constant and its size is $\tau_c$ samples (or channel uses) \cite[Ch.~2]{redbook},\cite[Ch.~2]{bjornson2017massive}. The channels change independently from one block to another according to a stationary ergodic random process. The number of samples per coherence interval is given by $\tau_c = T_{c} B_{c}$, according to the Nyquist-Shannon sampling theorem, where $T_{c}$ is the coherence time and $B_{c}$ is the coherence bandwidth.
In the proposed setup, ${\mathbf h}^{b,\rm c}_{b',k}\sim \CN({\mathbf{0}},\beta^{b,\rm c}_{b',k} {\mathbf I}_{M})$ is the channel response between the BS $b$ and CU $k$ in cell $b'$, where $\beta^{b,\rm c}_{b',k}$ is the corresponding large-scale fading and ${\mathbf h}^{b,\rm d}_{l} \sim \CN({\mathbf{0}},\beta^{b,\rm d}_l {\mathbf I}_M)$ is the channel response between D2D transmitter~$l$~and BS $b$ and $\beta^{b,\rm d}_l$ indicates the large-scale fading between D2D transmitter~$l$~and BS $b$. In addition, ${g}^{l,\rm c}_{b,k} \sim \CN(0,\beta^{l,\rm c}_{b,k}) $ denotes the channel response between CU~$k$~located in cell $b$ and D2D receiver~$l$, while ${g}^{l,\rm d}_{l'} \sim \CN(0,\beta^{l,\rm d}_{l'})$ denotes the channel between D2D transmitter~$l'$~and D2D receiver~$l$. In addition, $\beta^{l,\rm c}_{b,k}$ and $\beta^{l,\rm d}_{l'}$ refer to the large-scale fading between CU~$k$~located in cell $b$ and D2D receiver~$l$ and between D2D transmitter~$l'$~and D2D receiver~$l$, respectively.

In our network setup, it is assumed that the BSs and D2D receivers have no prior CSI at the beginning of a coherence interval; hence, channel estimation is carried out in each coherence interval. Therefore, the communication consists of two phases: UL pilot transmission and UL data transmission. To enable channel estimation, the CUs and D2D transmitters transmit pilot sequences of length $\tau$ in the pilot transmission phase and the remaining $\tau_{c} - \tau$ samples are utilized in the data transmission phase. In this setup, we can construct $\tau$ mutually orthogonal pilot sequences that are vectors of length $\tau$. It is assumed that we have $K$ orthogonal pilots for CUs that are reused in each cell and we have another set of $N$ orthogonal pilots designated for the $L\geq N$ D2D pairs. Hence, we have $\tau = N+K$. 

\subsection{Uplink data transmission}
The received signal during data transmission at BS $b$ is
\begin{equation}
\begin{aligned}\label{eq:data BS-cu}
{\mathbf y}^{\rm c}_{b}&= \sum\limits_{k=1}^{K}\sqrt{p^{\rm c}_{b,k}} {\mathbf h}^{b,\rm c}_{b,k} s^{\rm c}_{b,k}+\sum\limits_{\substack{b'=1,\\b'\neq b}}^{B}\sum\limits_{k=1}^{K}\sqrt{p^{\rm c}_{b',k}} {\mathbf h}^{b,\rm c}_{b',k} s^{\rm c}_{b',k}\\
&+\sum\limits_{l=1}^{L}\sqrt{{p^{\rm d}_{l}}}{\mathbf h}^{b,\rm d}_{l} s^{\rm d}_{l}+{\mathbf w}_{b}\\
\end{aligned}
\end{equation}
and the received signal at D2D receiver $l$ is 
\begin{equation}
\begin{aligned}\label{eq:data BS-dd}
{y}^{\rm d}_l&= \sum\limits_{b'=1}^{B}\sum\limits_{k=1}^{K}\sqrt{p^{\rm c}_{k}} {g}^{l,\rm c}_{b,k}s^{\rm c}_{b,k} +\sum\limits_{l'=1}^{L}\sqrt{{p^{\rm d}_{l'}}}{g}^{l,\rm d}_{l'}s^{\rm d}_{l'}+{w}_l,\\
\end{aligned}
\end{equation}
where~$p^{\rm c}_{b,k}$~and~$p^{\rm d}_{l'}$~are the transmit powers for data transmission of the CU~$k$ in cell~$b$~and the D2D transmitter~$l'$, respectively. ${\mathbf w}_{b}\sim\CN\left({\mathbf{0}},{\mathbf{I}}_M\right)$~and~${w}_l\sim\CN\left(0,1\right)$~indicate the normalized additive white Gaussian noise at BS $b$ and the $l$th D2D receiver, respectively. Also,~$s^{\rm c}_{b,k}$~and~$s^{\rm d}_{l'}$~denote zero mean and unit variance data symbols transmitted from CU~$k$~in cell~$b$~and D2D transmitter~$l'$, respectively.

\section{Analysis of spectral efficiency}
\label{Analysis of spectral efficiency}
In this section, we derive the SEs achieved when using the communication setup defined in Section~\ref{SystemModel}. Pilot contamination arises between CUs due to the pilot reuse between cells and pilot contamination also arises between the D2D pairs that are using the same pilot.
\subsection{Pilot transmission and channel estimation}
We denote the matrix of pilot sequences as $\mathbf{\boldsymbol{\Phi}} = [\mathbf{\boldsymbol{\Phi}}^{\rm c}~\mathbf{\boldsymbol{\Phi}}^{\rm d}]$ and it has size $\tau \times (K+N)$. The matrices $\mathbf{\boldsymbol{\Phi}}^{\rm c} = [\mathbf{\boldsymbol{\phi}}^{\rm c}_{1},\ldots,\mathbf{\boldsymbol{\phi}}^{\rm c}_{K}] \in \mathbb{C}^{\tau \times K}$ and $\mathbf{\boldsymbol{\Phi}}^{\rm d} = [\mathbf{\boldsymbol{\phi}}^{\rm d}_{1},\ldots,\mathbf{\boldsymbol{\phi}}^{\rm d}_{N}] \in \mathbb{C}^{\tau \times N}$ contain the orthogonal unit-norm pilot sequences assigned for CUs and D2D pairs, respectively. The pilot $\mathbf{\boldsymbol{\phi}}^{\rm c}_{k}$ is used by CU~$k$~in each of the cells.
The sets~$\mathcal{N}_1,\ldots,\mathcal{N}_N$~contain the indices of D2D pairs that are using pilots~$\mathbf{\boldsymbol{\phi}}^{\rm d}_1,\ldots,\mathbf{\boldsymbol{\phi}}^{\rm d}_N$, respectively.
The received pilot signal~${\mathbf Y}^{b} \in \mathbb{C}^{M \times \tau}$~at BS~$b$~from all CUs and D2D transmitters is~\cite[Ch.~3]{redbook},\cite[Ch.~3]{bjornson2017massive}
\begin{equation}
\begin{aligned}\label{eq:pilot BS-DD1}
{\mathbf Y}^{b}&= \sum\limits_{k=1}^{K}\sum_{b'=1}^{B}\sqrt{\tau {{p^{\rm p,\rm c}_{b',k}}}}{\mathbf h}^{b,\rm c}_{b',k} \left(\boldsymbol{\phi}^{{\rm c}}_k\right)^{\rm H} \\
&+ \sum\limits_{i=1}^{N}\sum_{l \in \mathcal{N}_i}\sqrt{\tau {p^{\rm p,\rm d}_{l}}}{\mathbf h}^{b,\rm d}_{l} \left(\boldsymbol{\phi}^{{\rm d}}_i\right)^{\rm H}+{\mathbf W}^{b}, \\
\end{aligned}
\end{equation}
where~$p^{\rm p,c}_{b,k}$~denotes the pilot transmit power used by CU~$k$~in cell~$b$~and~$p^{\rm p,d}_{l}$~is the transmit power of D2D transmitter~$l$~for pilot transmission. ${\mathbf W}^{b} \in \mathbb{C}^{M\times \tau}$~is the normalized additive white Gaussian noise at the BS~$b$~which has independent entries having the distribution~$\CN\left(0,1\right)$. Each BS multiplies the received signal matrix in the pilot transmission phase with each of the pilot signals to despread the signals. Hence, the received pilot signals at BS~$b$~from its CUs after despreading by $\mathbf{\boldsymbol{\Phi}}^{\rm c}$ is~\cite[Ch.~3]{redbook},\cite[Ch.~3]{bjornson2017massive}
\begin{equation}
\begin{aligned}\label{eq:pilot BS-cu}
{\mathbf Y}^{b} \boldsymbol{\Phi}^{\rm c} &= \sum\limits_{k=1}^{K}\sum_{b'=1}^{B}\sqrt{\tau {p^{\rm p,\rm c}_{b',k}}}{\mathbf h}^{b,\rm c}_{b',k} \left(\boldsymbol{\phi}^{{\rm c}}_k\right)^{\rm H} \boldsymbol{\Phi}^{\rm c} +{\mathbf W''}^{b,\rm c},\\
\end{aligned}
\end{equation}
where~${{\mathbf W}''}^{b,\rm c} = {\mathbf W}^{b }\boldsymbol{\Phi}^{\rm c} \in \mathbb{C}^{M\times K}$ is the normalized additive white Gaussian noise at the BS~$b$~with independent entries having the distribution~$\CN\left(0,1\right)$. The received pilot signal at BS $b$ from all D2D transmitters is the second part of \eqref{eq:pilot BS-DD1} and despreading by multiplication with $\mathbf{\boldsymbol{\Phi}}^{\rm d}$ yields
\begin{equation}
\begin{aligned}\label{eq:pilot BS-DD}
{\mathbf Y}^{b} {\pmb\Phi}^{\rm d}&= \sum\limits_{i=1}^{N}\sum_{l\in \mathcal{N}_i}\sqrt{\tau {p^{\rm p,\rm d}_{l}}}{\mathbf h}^{b,\rm d}_{l}(\pmb{\phi}^{{\rm d}}_i)^{\rm H}\pmb{\Phi}^{{\rm d}}+{\mathbf W'}^{b,\rm d},\\
\end{aligned}
\end{equation}
where~${\mathbf W'}^{b,\rm d} = {\mathbf W}^{b,\rm d}\boldsymbol{\Phi}^{\rm d}\in \mathbb{C}^{M\times N}$~is the normalized additive white Gaussian noise at the BS~$b$~which has independent entries having the distribution~$\CN\left(0,1\right)$.  
The minimum mean-squared error (MMSE) estimates of the channel vectors from CU~$k$ at BS $b$ are~\cite{kailath2000linear,kay1993fundamentals}
\begin{equation}
\begin{aligned}\label{eq:mmse BS-cu}
\hat{{\mathbf h}}^{b,\rm c}_{b,k} =  \frac{\sqrt{\tau p^{\rm p,c}_{b,k}} \beta^{b,\rm c}_{b,k}}{1+ \tau \sum\limits_{b'=1}^{B} p^{\rm p,c}_{b',k} \beta^{b,\rm c}_{b',k}} {{\mathbf Y}^{b} \mathbf{\boldsymbol{\phi}}^{\rm c}_{k}}.\\
\end{aligned}
\end{equation}
The MMSE estimate at BS~$b$~of the sum of the channel vectors~$\sum_{l\in \mathcal{N}_i} \sqrt{\smash[b]{\tau p^{\rm p, d}_{l}}} {\mathbf h}^{b,\rm d}_{l}$~from D2D transmitters in set~$\mathcal{N}_i$~from D2D transmitters in set~$\mathcal{N}_i$~is 
\begin{equation}
\begin{aligned} \label{eq:mmse BS-dd-ni}
\hat{{\mathbf h}}^{b,\mathcal{N}_i}_{} = \frac{\sum\limits_{l' \in \mathcal{N}_i} {\tau p^{\rm p, d}_{l'}} \beta^{b,\rm d}_{l'}}{1+ \sum\limits_{l'\in \mathcal{N}_i}\tau {p^{\rm p,\rm d}_{l'}} \beta^{b,\rm d}_{l'}} {{\mathbf Y}^{b}} \mathbf{\boldsymbol{\phi}}^{d}_{i}.\\
\end{aligned}
\end{equation}
In addition, the MMSE estimates of the channel vectors from  D2D transmitter~$l \in \mathcal{N}_i$~at BS~$b$~are
\begin{equation}
\begin{aligned}\label{eq:mmse BS-dd}
\hat{{\mathbf h}}^{b,\rm d}_{l} = \frac{\sqrt{\tau {p^{\rm p, d}_{l}}} \beta^{b,\rm d}_{l}}{1+ \sum\limits_{l'\in \mathcal{N}_i}\tau {p^{\rm p,\rm d}_{l'}} \beta^{b,\rm d}_{l'}}  {\mathbf Y}^{b} \mathbf{\boldsymbol{\phi}}^{\rm d}_i.\\
\end{aligned}
\end{equation}
The mean square of the channel estimates at the BS are~$\mathbb{E}[\|\hat{{\mathbf h}}^{b,\rm c}_{b,k}\|^2]= \gamma^{b,\rm c}_{b,k} M$,~$\mathbb{E}[\|	\hat{{\mathbf h}}^{b,\mathcal{N}_i}\|^2]= \gamma^{b,\mathcal{N}_i} M$, and~$\mathbb{E}[\|\hat{{\mathbf h}}^{b,\rm d}_{l}\|^2]=\gamma^{b,\rm d}_{l} M$, 
in which we have used
\begin{equation}
\begin{aligned}
\gamma^{b,\rm c}_{b,k} = \frac{\tau p^{\rm p,c}_{b,k} \left({\beta^{b,\rm c}_{b,k}}\right)^2}{1+\tau \sum\limits_{b'=1}^{B} p^{\rm p,c}_{b',k} \beta^{b,\rm c}_{b',k}},
\end{aligned}
\end{equation}
\begin{equation}
\begin{aligned}
\gamma^{b,\mathcal{N}_i} = \frac{\tau \left(\sum\limits_{l' \in \mathcal{N}_i} \sqrt{{p^{\rm p,\rm d}_{l'}}} {\beta^{b,\rm d}_{l'}}\right)^2}{1+\sum\limits_{l'\in \mathcal{N}_i}\tau p^{\rm p, d}_{l'} \beta^{b,\rm d}_{l'}},
\end{aligned}
\end{equation}
\begin{equation}
\begin{aligned}
\gamma^{b,\rm d}_{l} = \frac{\tau {p^{\rm p,\rm d}_{l}} \left({\beta^{b,\rm d}_{l}}\right)^2}{1+\sum\limits_{l'\in \mathcal{N}_i}\tau p^{\rm p, d}_{l'} \beta^{b,\rm d}_{l'}}.
\end{aligned}
\end{equation}
Note that the channel estimate~$\hat{{\mathbf h}}^{b,\rm d}_{l}$~of  D2D transmitter~$l \in \mathcal{N}_i$~is a scaled version of~$\hat{{\mathbf h}}^{b,\mathcal{N}_i}$, which is a property that we will utilize later in the receive processing. The D2D receivers apply the same procedure for channel estimation. Therefore, the received signals at D2D receiver~$l$~from the pilot transmission of CU~$k$ located in cell~$b$~and D2D transmitter~$l' \in \mathcal{N}_i$~after despreading the signals are
\begin{equation}
\begin{aligned}\label{eq:pilot DD-cu}
{y}^{l,\rm c}_{b,k}&=  \sum\limits_{b'=1}^{B}\sqrt{\tau p^{l,\rm p,c}_{b',k}} {g}^{l,\rm c}_{b',k}+{w}^{\rm c}_{l}, \\
\end{aligned}
\end{equation}
\begin{equation}
\begin{aligned}\label{eq:pilot dd}
{y}^{l,\rm d}_{l'}&= \sum\limits_{l''\in n_{i}}\sqrt{\tau {p^{\rm p,d}_{ l''}}}{g}^{l,\rm d}_{l''} +{ w}^{\rm d}_l, \\
\end{aligned}
\end{equation}
where~${w}^{\rm d}_l\sim\CN\left(0,1\right)$~and~${ w}^{\rm c}_l \sim\CN\left(0,1\right)$~are the normalized additive white Gaussian noise terms at D2D receiver~$l$.  
The pilot transmission of all transmitters, i.e., cellular and D2D, are used for channel estimation and the MMSE estimates of the channels from CU~$k$~located in cell $b$ and from D2D transmitter~$l' \in \mathcal{N}_i$~are respectively given by
\begin{equation}
\begin{aligned}\label{eq:mmse dd-cu}
\hat{g}^{l,\rm c}_{b,k} = \frac{\sqrt{\tau p^{l,\rm p,c}_{b,k}} \beta^{l,\rm c}_{b,k}}{1+ \sum\limits_{b'=1}^{B}\tau p^{l,\rm p,c}_{b',k} \beta^{l,\rm c}_{b',k}} {y}^{l,\rm c}_{b,k}, \\
\end{aligned}
\end{equation}
\begin{equation}
\begin{aligned}\label{eq:mmse dd}
\hat{g}^{l,\rm d}_{l'} = \frac{\sqrt{\tau p^{\rm p,d}_{l'}} \beta^{l,\rm d}_{l'}}{1+ \sum\limits_{l''\in \mathcal{N}_i}\tau p^{\rm p,d}_{l''} \beta^{l,\rm d}_{l''}} {y}^{l,\rm d}_{l'}.\\
\end{aligned}
\end{equation}
In addition, the mean square of the channel estimates at D2D receiver~$l$~are
\begin{equation}
\begin{aligned}
\mathbb{E}\left[\left|\hat{g}^{l,\rm c}_{b,k}\right|^2\right] = {\gamma^{l,\rm c}_{b,k}}=  \frac{\tau p^{l,\rm p,c}_{b,k} \left({\beta^{l,\rm c}_{b,k}}\right)^2}{1+ \sum\limits_{b'=1}^{B}\tau p^{l,\rm p,c}_{b',k} \beta^{l,\rm c}_{b',k}},
\end{aligned}
\end{equation}
\begin{equation}
\begin{aligned}
\mathbb{E}\left[\left|\hat{g}{^{l,\rm d}_{l'}}\right|^2\right] = \gamma^{l,\rm d}_{l'}=  \frac{\tau p^{\rm p,d}_{l'} \left({\beta^{l,\rm d}_{l'}}\right)^2}{1+ \sum\limits_{l''\in \mathcal{N}_i}\tau p^{\rm p,d}_{l''} \beta^{l,\rm d}_{l''}}.
\end{aligned}
\end{equation}
	
\begin{figure*}
	\begin{equation}\tag{18}
	\begin{aligned}\label{eq:MRC BS1}
	\left(\hat{\mathbf h}^{b,\rm c}_{b,k}\right)^{\rm H}{\mathbf y}^{\rm c}_{b}&= \sqrt{p^{\rm c}_{b,k}}\left(\hat{\mathbf h}^{b,\rm c}_{b,k}\right)^{\rm H}{\mathbf h}{^{b,\rm c}_{b,k}} s^{\rm c}_{b,k}+ \sum\limits_{\substack{k'=1,\\ k' \neq k}}^{K}\sqrt{p^{\rm c}_{b,k'}}\left(\hat{\mathbf h}^{b,\rm c}_{b,k}\right)^{\rm H} {\mathbf h}{^{b,\rm c}_{b,k'}}  s^{\rm c}_{b,k'}+\\
	&\sum\limits_{\substack{b'=1,\\b'\neq b}}^{B}\sum\limits_{k'=1}^{K} \sqrt{p^{\rm c}_{b',k}} \left(\hat{\mathbf h}^{b,\rm c}_{b,k}\right)^{\rm H} {\mathbf h}^{b,\rm c}_{b',k'} s^{\rm c}_{b',k'}+ \sum\limits_{l=1}^{L}\sqrt{p^{\rm d}_{l}}\left(\hat{\mathbf h}^{b,\rm c}_{b,k}\right)^{\rm H}{\mathbf h}^{b,\rm d}_{l} s^{\rm d}_{l}+\left(\hat{\mathbf h}^{b,\rm c}_{b,k}\right)^{\rm H}{\mathbf w}_{b}\\
	\end{aligned}
	\end{equation}
	\hrulefill
	\vspace{-3mm}
\end{figure*}
\subsection{Spectral efficiency with MR processing}
In this subsection, we assume that MR processing is applied at all BSs to detect the signals of their own users. The received data signal of user~$k$~at the cell~$b$~is expressed as \eqref{eq:MRC BS1}, as shown at the top of the next page, where the first term is the desired part of the received signal, the second term is cellular intracell interference, the third term is intercell interference, the fourth term is D2D interference, and the last term is noise. We rewrite the received data signal of the CU~$k$~at the cell~$b$~as
\setcounter{equation}{18}
\begin{equation}
\begin{aligned}\label{eq:MRC BS}
\left(\hat{\mathbf h}^{b,\rm c}_{b,k}\right)^{\rm H}{\mathbf y}^{\rm c}_{b}&= \sqrt{p^{\rm c}_{b,k}} \mathbb{E}\left[\left(\hat{\mathbf h}^{b,\rm c}_{b,k}\right)^{\rm H}{\mathbf h}^{b,\rm c}_{b,k}\right] s^{\rm c}_{b,k} \\
&+\sqrt{p^{\rm c}_{b,k}}\left(\left(\hat{\mathbf h}^{b,\rm c}_{b,k}\right)^{\rm H}{\mathbf h}^{b,\rm c}_{b,k} - \mathbb{E}\left[\left(\hat{\mathbf h}^{b,\rm c}_{b,k}\right)^{\rm H}{\mathbf h}^{b,\rm c}_{b,k}\right]\right) s^{\rm c}_{b,k}\\
&+ \sum\limits_{\substack{k'=1,\\ k' \neq k}}^{K}\sqrt{p^{\rm c}_{b,k'}}\left(\hat{\mathbf h}^{b,\rm c}_{b,k}\right)^{\rm H} {\mathbf h}{^{b,\rm c}_{b,k'}}  s^{\rm c}_{b,k'}+\\&\sum\limits_{\substack{b'=1,\\b'\neq b}}^{B}\sum\limits_{k'=1}^{K} \sqrt{p^{\rm c}_{b',k}} \left(\hat{\mathbf h}^{b,\rm c}_{b,k}\right)^{\rm H} {\mathbf h}^{b,\rm c}_{b',k'} s^{\rm c}_{b',k'}\\
&+ \sum\limits_{l=1}^{L}\sqrt{p^{\rm d}_{l}}\left(\hat{\mathbf h}^{b,\rm c}_{b,k}\right)^{\rm H}{\mathbf h}^{b,\rm d}_{l} s^{\rm d}_{l}+\left(\hat{\mathbf h}^{b,\rm c}_{b,k}\right)^{\rm H}{\mathbf w}_{b},\\
\end{aligned}
\end{equation}
where the first term is treated as the desired part of the received signal and the rest of the terms are treated as noise in the signal detection. More precisely, we utilize the use-and-then-forget technique~\cite[Ch.~3]{redbook}~to lower bound the capacity of each of the CUs, using the capacity bound for a deterministic channel with additive white non-Gaussian noise provided in~\cite[Sec.~2.3]{redbook}. We get the lower bound on the capacity of CU~$k$~in cell~$b$ as 
	\begin{equation}
	\begin{aligned}\label{eq:MRC-bound-pilot}
	R^{b,\rm c}_{b,k} =&  
	\left(1 - \frac{\tau }{\tau_c}\right)\log_2 \left(1+\frac{M p^{\rm c}_{b,k} \left(\frac{\tau p^{\rm p,c}_{b,k} \left({\beta^{b,\rm c}_{b,k}}\right)^2}{1+\tau \sum\limits_{b'=1}^{B} p^{\rm p,c}_{b',k} \beta^{b,\rm c}_{b',k}}\right)}{ I^{b,\rm c}_{b,k}\left({\rm MR}\right) }\right),\\
	\end{aligned}
	\end{equation}
where $I^{b,\rm c}_{b,k}\left({\rm MR}\right)$ is defined as seen in \eqref{eq:MRC-bound-pilot_I(MR)} at the top of the next page.
\begin{figure*}
	\begin{equation}\label{eq:MRC-bound-pilot_I(MR)}
	\begin{aligned}
	I^{b,\rm c}_{b,k}\left({\rm MR}\right)=& 1+ \sum\limits_{b'=1}^{B}\sum\limits_{k'=1}^{K} p^{\rm c}_{b',k'} \beta^{b,\rm c}_{b',k'} + \sum\limits_{l=1}^{L}p^{\rm d}_{l}  \beta^{b,\rm d}_l + M\sum\limits_{b'=1,b'\neq b}^{B} p^{\rm c}_{b',k} \left(\frac{\tau p^{\rm p,c}_{b',k} \left({\beta^{b,\rm c}_{b',k}}\right)^2}{1+\tau \sum\limits_{b''=1}^{B} p^{\rm p,c}_{b'',k} \beta^{b,\rm c}_{b'',k}}\right),\\
	\end{aligned}
	\end{equation}
	\hrulefill
\end{figure*}

We can write \eqref{eq:MRC-bound-pilot} as
	\begin{equation}
	\begin{aligned}\label{eq:MRC-pilot power}
	R^{b,\rm c}_{b,k}= &  
	\left(1 - \frac{\tau }{\tau_c}\right)\log_2 \left(1+\underset{{\rm SINR}^{b, \rm c}_{b,k}}{\underbrace{\frac{M \tau p^{\rm c}_{b,k} p^{\rm p,c}_{b,k} \left({\beta^{b,\rm c}_{b,k}}\right)^2 }{{I'}^{b,\rm c}_{b,k}\left({\rm MR}\right) }}}\right)\\
	\end{aligned}
	\end{equation}
where ${I'}^{b,\rm c}_{b,k}\left({\rm MR}\right)$ is defined in \eqref{eq:MRC-bound-pilot_I(MR)_prime}, shown at the top of the next page.
\begin{figure*}
	\begin{equation}\label{eq:MRC-bound-pilot_I(MR)_prime}
	\begin{aligned}
	{I'}^{b,\rm c}_{b,k}\left({\rm MR}\right)=& \left({1+\tau \sum\limits_{b'=1}^{B} p^{\rm p,c}_{b',k} \beta^{b,\rm c}_{b',k}}\right) \left( 1+ \sum\limits_{b'=1}^{B}\sum\limits_{k'=1}^{K} p^{\rm c}_{b',k'} \beta^{b,\rm c}_{b',k'} + \sum\limits_{l=1}^{L}p^{\rm d}_{l}  \beta^{b,\rm d}_l \right)+M\tau \sum\limits_{b'=1,b'\neq b}^{B} p^{\rm c}_{b',k} p^{\rm p,c}_{b',k} \left({\beta^{b,\rm c}_{b',k}}\right)^2.
	\end{aligned}
	\end{equation}
	\hrulefill\
\end{figure*}
\subsection{Spectral efficiency with zero-forcing processing}
In this subsection, we consider ZF processing at all BSs, instead of MR processing, to suppress interference between the users. In addition to suppressing intracell interference, we also suppress interference from the D2D transmitters. For ease of notation, we rewrite~\eqref{eq:data BS-cu}~in matrix form as
\begin{equation}
\begin{aligned}\label{eq:data BS-cu-matrix}
	{\mathbf y}^{\rm c}_{b}&= {\mathbf{H}}^{b,\rm c}_{b} {\mathbf{D}}^{1/2}_{{\mathbf{p}}^{\rm c}_{b}}{\mathbf{s}}^{\rm c}_{b}+\sum\limits_{\substack{b'=1,\\b'\neq b}}^{B} {\mathbf{H}}^{b,\rm c}_{b'} {\mathbf{D}}^{1/2}_{{\mathbf{p}}^{\rm c}_{b'}}{\mathbf{s}}^{\rm c}_{b'} + {\mathbf{H}}^{b,\rm d} {\mathbf{D}}^{1/2}_{\mathbf{p}^{\rm d}} {\mathbf{s}}^{b,\rm d} +{\mathbf{w}}_{b}
\end{aligned}
\end{equation}
where~${\mathbf{H}}^{b,\rm c}_{b'} = [{\mathbf h}^{b,\rm c}_{b',1},\ldots,{\mathbf h}^{b',\rm c}_{b,K}]$~is the channel matrix of the~$K$~CUs located in cell~$b'$~to BS~$b$~and~${\mathbf{H}}^{b',\rm d} =[{\mathbf h}^{b',\rm d}_{1},\ldots,{\mathbf h}^{b',\rm d}_{L}]$~is the channel matrix of the D2D transmitters to BS~$b'$. ${\mathbf{D}}_{{\mathbf{p}}^{\rm c}_{b'}}$~and~${\mathbf{D}}_{{\mathbf p}^{\rm d}}$~are diagonal matrices indicating the transmit power of CUs at cell~$b'$~and D2D transmitters, respectively, where ${\mathbf{p}}^{\rm c}_{b'} = [p^{\rm c}_{b'1},\ldots,p^{\rm c}_{b'K}]^{T}$ and ${\mathbf{p}}^{\rm d} = [ p^{\rm d}_{1},\ldots,p^{\rm d}_{L}]^{T}$~are the vectors on the diagonals. The channel matrices in~\eqref{eq:data BS-cu-matrix}~can be replaced with the corresponding estimated channel matrices $\hat{\mathbf{H}}^{b,\rm c}_{b'} = [\hat{{\mathbf h}}^{b,\rm c}_{b'1},\ldots,\hat{{\mathbf h}}^{b,\rm c}_{b',K}]$~and~$\hat{\mathbf{H}}^{b,\rm d} = [\hat{{\mathbf h}}^{\rm d}_1,\ldots,\hat{{\mathbf h}}^{\rm d}_{L}]$~and estimation error matrices defined as~$\tilde{\mathbf{H}}^{b,\rm c}_{b'}= \hat{\mathbf{H}}^{b,\rm c}_{b'} - {\mathbf{H}}^{b,\rm c}_{b'} = [\tilde{{\mathbf h}}^{b,\rm c}_{b'1},\ldots,\tilde{{\mathbf h}}^{b,\rm c}_{b',K}]$~and~$\tilde{\mathbf{H}}^{b,\rm d} = \hat{\mathbf{H}}^{b,\rm d} - {\mathbf{H}}^{b,\rm d}= [\tilde{{\mathbf h}}^{b,\rm d}_1,\ldots,\tilde{{\mathbf h}}^{b,\rm d}_{L}]$. Hence, \eqref{eq:data BS-cu-matrix}~is rewritten as 
\begin{equation}
\begin{aligned}\label{eq:data BS-cu matrix}
	{\mathbf y}^{\rm c}_{b}&= \hat{\mathbf{H}}^{b,\rm c}_{b} {\mathbf{D}}^{1/2}_{{\mathbf{p}}^{\rm c}_{b}}{\mathbf{s}}^{\rm c}_{b}+\sum\limits_{\substack{b'=1,\\b'\neq b}}^{B} \hat{\mathbf{H}}^{b,\rm c}_{b'} {\mathbf{D}}^{1/2}_{{\mathbf{p}}^{\rm c}_{b'}}{\mathbf{s}}^{\rm c}_{b'} + \hat{\mathbf{H}}^{b,\rm d} {\mathbf{D}}^{1/2}_{\mathbf{p}^{\rm d}} {\mathbf{s}}^{b,\rm d}\\
	&+{\mathbf{w}}_{b} -\tilde{\mathbf{H}}^{b,\rm c}_{b} {\mathbf{D}}^{1/2}_{{\mathbf{p}}^{\rm c}_{b}}{\mathbf{s}}^{\rm c}_{b} -\!\!\!\sum\limits_{\substack{b'=1,\\b'\neq b}}^{B} \!\!\!\tilde{\mathbf{H}}^{b,\rm c}_{b'} {\mathbf{D}}^{1/2}_{{\mathbf{p}}^{\rm c}_{b'}}{\mathbf{s}}^{\rm c}_{b'}-\tilde{\mathbf{H}}^{b,\rm d} {\mathbf{D}}^{1/2}_{\mathbf{p}^{\rm d}} {\mathbf{s}}^{b,\rm d}.
\end{aligned}
\end{equation}
As mentioned previously, the channel estimates at each BS~$b$~for D2D transmitters that use the same pilot are parallel vectors. Hence, the collection of~$N$~vectors $\hat{\mathbf{h}}^{b,n_1},\ldots,\hat{\mathbf{h}}^{b,n_N}$~from~\eqref{eq:mmse BS-dd-ni} contains scaled version of all the~$L$~channel estimates~$\hat{{\mathbf h}}^{b,\rm d}_{1}, \ldots,\hat{{\mathbf h}}^{b,\rm d}_{L}$. Therefore, we construct the~$M \times (K+N)$~matrix 
\begin{equation}
\begin{aligned}
	\hat{\mathbf{H}}^{b}_{b} &= \left[\hat{\mathbf{H}}^{b,\rm c}_{b}\quad\hat{\mathbf{h}}^{b,\mathcal{N}_1}\,\,\ldots\,\,\hat{\mathbf{h}}^{b,\mathcal{N}_N}\right]
\end{aligned}
\end{equation}
to describe all the channel directions that are relevant for interference mitigation at BS $b$.
We can write this matrix as
\begin{equation}
\begin{aligned}
	\hat{\mathbf{H}}^{b}_{b} &= {\mathbf{Z}} {\mathbf{D}}^{1/2}_{\bgamma^{b}}
\end{aligned}
\end{equation}
where the matrix~${\mathbf Z}$~has i.i.d.~$\CN\left(0,1\right)$~elements and ${\mathbf{D}}_{\bgamma^{b}}$~is a diagonal matrix with~${{\bgamma}^{b}} = [\gamma^{b,\rm c}_{b,1},\ldots,\gamma^{b,\rm c}_{b,K},\gamma^{b,n_1},\ldots,\gamma^{b,n_N}]$~at the diagonal.
By using this matrix, we define the ZF detection matrix at cell~$b$~as 
\begin{equation}
\begin{aligned}
	{\mathbf V}^{b}&= \hat{\mathbf{H}}^{b}_{b} \left(\left({\hat{\mathbf{H}}_{b}}^{b}\right)^{\rm H} \hat{\mathbf{H}}^{b}_{b} \right)^{-1} {\mathbf{D}}^{1/2}_{\bgamma^{b}} = {\mathbf{Z}}\left({\mathbf{Z}}^{\rm H} {\mathbf{Z}}\right)^{-1}.
\end{aligned}
\end{equation}
By applying this decoder matrix in cell~$b$, the received signal becomes as seen in \eqref{eq:ZF BS} at the top of the next page
\begin{figure*}[t!]
	\begin{equation}
	\begin{aligned}\label{eq:ZF BS}
		({\mathbf V}^{b})^{\rm H}{\mathbf y}^{\rm c}_{b} &= \left({\mathbf{Z}}^{\rm H} {\mathbf{Z}} \right)^{-1} {\mathbf{Z}}^{\rm H} \hat{\mathbf{H}}^{b,\rm c}_{b} {\mathbf{D}}^{1/2}_{{\mathbf{p}}^{\rm c}_{b}}{\mathbf{s}}^{\rm c}_{b}+	\left({\mathbf{Z}}^{\rm H} {\mathbf{Z}} \right)^{-1} {\mathbf{Z}}^{\rm H} \sum\limits_{\substack{b'=1,\\b'\neq b}}^{B} \hat{\mathbf{H}}^{b,\rm c}_{b'} {\mathbf{D}}^{1/2}_{{\mathbf{p}}^{\rm c}_{b'}}{\mathbf{s}}^{\rm c}_{b'} +\left({\mathbf{Z}}^{\rm H} {\mathbf{Z}} \right)^{-1} {\mathbf{Z}}^{\rm H} \hat{\mathbf{H}}^{b,\rm d} {\mathbf{D}}^{1/2}_{\mathbf{p}^{\rm d}} {\mathbf{s}}^{b,\rm d} \\
		&+\left({\mathbf{Z}}^{\rm H} {\mathbf{Z}} \right)^{-1} {\mathbf{Z}}^{\rm H}\left({\mathbf{w}}_{b} -\tilde{\mathbf{H}}^{b,\rm c}_{b} {\mathbf{D}}^{1/2}_{{\mathbf{p}}^{\rm c}_{b}}{\mathbf{s}}^{\rm c}_{b} -\sum\limits_{\substack{b'=1,\\b'\neq b}}^{B} \tilde{\mathbf{H}}^{b,\rm c}_{b'} {\mathbf{D}}^{1/2}_{{\mathbf{p}}^{\rm c}_{b'}}{\mathbf{s}}^{\rm c}_{b'}-\tilde{\mathbf{H}}^{b,\rm d} {\mathbf{D}}^{1/2}_{\mathbf{p}^{\rm d}} {\mathbf{s}}^{b,\rm d}\right)
	\end{aligned}
	\end{equation}
	\hrulefill
\end{figure*}
and the received signal of CU user~$k$~at the cell~$b$~is provided in \eqref{eq:ZF BS user}, that can be seen at the top of the next page,
\begin{figure*}
	\begin{equation}\label{eq:ZF BS user}
	\begin{split}
	\begin{aligned}
		&\left[({\mathbf V}^{b})^{\rm H}{\mathbf y}^{\rm c}_{b}\right]_k= \sqrt{\gamma^{b,\rm c}_k p^{\rm c}_{b,k}} {s}^{\rm c}_{b,k} + \sum\limits_{\substack{b'=1,\\ b' \neq b}} \sqrt{\gamma^{b,\rm c}_{b',k} p^{\rm c}_{b',k}} {s}^{\rm c}_{b',k} + \left[\left({\mathbf{Z}}^{\rm H} {\mathbf{Z}} \right)^{-1} {\mathbf{Z}}^{\rm H}\left({\mathbf{w}}_{b} -\tilde{\mathbf{H}}^{b,\rm c}_{b} {\mathbf{D}}^{1/2}_{{\mathbf{p}}^{\rm c}_{b}}{\mathbf{s}}^{\rm c}_{b} -\sum\limits_{\substack{b'=1,\\b'\neq b}}^{B} \tilde{\mathbf{H}}^{b,\rm c}_{b'} {\mathbf{D}}^{1/2}_{{\mathbf{p}}^{\rm c}_{b'}}{\mathbf{s}}^{\rm c}_{b'}-\tilde{\mathbf{H}}^{b,\rm d} {\mathbf{D}}^{1/2}_{\mathbf{p}^{\rm d}} {\mathbf{s}}^{b,\rm d}\right)\right]_k
	\end{aligned}
	\end{split}
	\end{equation}
	\hrulefill
\end{figure*} 
where the first term is the desired signal and the rest is interference plus noise and its variance is provided in \eqref{eq:varzf} can be seen at the top of the next page.
\begin{figure*}[t!]
	\begin{equation}
	\begin{aligned}\label{eq:varzf}
		&\var\left\{\left[\left({\mathbf{Z}}^{\rm H} {\mathbf{Z}} \right)^{-1} {\mathbf{Z}}^{\rm H}\left({\mathbf{w}}_{b} -\tilde{\mathbf{H}}^{b,\rm c}_{b} {\mathbf{D}}^{1/2}_{{\mathbf{p}}^{\rm c}_{b}}{\mathbf{s}}^{\rm c}_{b} -\sum\limits_{\substack{b'=1,\\b'\neq b}}^{B} \tilde{\mathbf{H}}^{b,\rm c}_{b'} {\mathbf{D}}^{1/2}_{{\mathbf{p}}^{\rm c}_{b'}}{\mathbf{s}}^{\rm c}_{b'}-\tilde{\mathbf{H}}^{b,\rm d} {\mathbf{D}}^{1/2}_{\mathbf{p}^{\rm d}} {\mathbf{s}}^{b,\rm d}\right)\right]_k\right\}  \\
		& =\left(1+ \sum\limits_{b'=1}^{B}\sum\limits_{k'=1}^{K} p^{\rm c}_{b',k'} \left(\beta^{b,\rm c}_{b',k'}-\gamma^{b,\rm c}_{b',k'}\right) + \sum\limits_{l=1}^{L}p^{b,\rm d}_{l} \left(\beta^{b,\rm d}_l -\gamma^{b,\rm d}_l\right)\right) \mathbb{E}\left\{\left[{\mathbf{Z}}^{\rm H} {\mathbf{Z}} \right]^{-1}_{k,k} \right\}\\
		& =\left(1+ \sum\limits_{b'=1}^{B}\sum\limits_{k'=1}^{K} p^{\rm c}_{b',k'} \left(\beta^{b,\rm c}_{b',k'}-\gamma^{b,\rm c}_{b',k'}\right) + \sum\limits_{l=1}^{L}p^{b,\rm d}_{l} \left(\beta^{b,\rm d}_l -\gamma^{b,\rm d}_l\right)\right)\frac{1}{M-\left(K+N\right)}.
	\end{aligned}
	\end{equation}
	\hrulefill
\end{figure*}
Note that this ZF matrix will suppress interference not only between CUs, as in conventional massive MIMO, but also from D2D transmitters. By using the capacity bounding technique in \cite[Sec.~2.3]{redbook}, we have the lower bound on the capacity of the CU~$k$~in cell~$b$ as 
%
\begin{equation}
\begin{aligned}\label{eq:ZF-pilot power}
R^{b,\rm c}_{b,k} & =  \left(1 - \frac{\tau }{\tau_{c}}\right) \times \\
&\log_2 \left(1+\underset{{\rm SINR}^{b,\rm c}_{k,b}}{\underbrace{\frac{p^{\rm c}_{b,k} \left({\tau p^{\rm p,c}_{b,k} \left({\beta^{b,\rm c}_{b,k}}\right)^2}\right) \left(M-\left(K+N\right)\right) }{I^{b,\rm c}_{b,k}\left({\rm ZF}\right)}}} \right),\\
\end{aligned}
\end{equation}
where $I^{b,\rm c}_{b,k}({\rm ZF})$ is defined in \eqref{eq:ZF-pilot power-denominator} that can be seen at the bottom of the next page.
\begin{figure*}[]
\begin{equation}
\begin{aligned}\label{eq:ZF-pilot power-denominator}
& I^{b,\rm c}_{b,k}\left({\rm ZF}\right) =1+\tau \sum\limits_{b'=1}^{B} p^{\rm p,c}_{b',k} \beta^{b,\rm c}_{b',k} +  \left(1+\tau \sum\limits_{b'=1}^{B} p^{\rm p,c}_{b',k} \beta^{b,\rm c}_{b',k}\right)\sum\limits_{b''=1}^{B}\sum\limits_{k''=1}^{K} \left(\frac{1+\tau \sum\limits_{b'\neq b''}^{B} p^{\rm p,c}_{b',k''} \beta^{b,\rm c}_{b',k''}}{1+\tau \sum\limits_{b'=1}^{B} p^{\rm p,c}_{b',k''} \beta^{b,\rm c}_{b',k''}}\right) p^{\rm c}_{b'',k''}\beta^{b,\rm c}_{b'',k''}\\
&+\left(1+\tau \sum\limits_{b'=1}^{B} p^{\rm p,c}_{b',k} \beta^{b,\rm c}_{b',k}\right) \sum\limits_{l=1}^{L}  p^{b,\rm d}_{l}\beta^{b,\rm d}_{l} \left(\frac{1 + \sum\limits_{l'\in \mathcal{N}_i,l'\neq l} \tau p^{\rm p,\rm d}_{l'} \beta^{b,d}_{l'}}{1+\sum\limits_{l'\in \mathcal{N}_i}\tau p^{\rm p, d}_{l'} \beta^{b,\rm d}_{l'}} \right)	+\left(M-\left(K+N\right)\right) \sum\limits_{\substack{b''=1,\\ b'' \neq b}} \tau  p^{\rm c}_{b'',k} p^{\rm p,c}_{b'',k} \left({\beta^{b,\rm c}_{b'',k}}\right)^2.
\end{aligned}
\end{equation}
\hrulefill
\end{figure*}
\subsection{Spectral efficiency of D2D communication}
To calculate the SE of the D2D transmissions, we use the MMSE estimates provided in~\eqref{eq:mmse dd-cu}~and~\eqref{eq:mmse dd}. To detect the desired data from D2D transmitter~$l$, D2D receiver~$l$~multiplies the received signal in~\eqref{eq:data BS-dd}~with the complex conjugate of its channel estimate~$\hat{g}^{l,\rm d}_{l}$. In addition, D2D receiver~$l$~uses the channel estimates~$\hat{g}^{l,\rm c}_{b,k}~\forall b, k$~and~$\hat{g}^{l',\rm d}_{l}~\forall l'$~as side information during data detection. We define~$\Omega =\{\hat{g}^{l,\rm d}_{l'}\}^{L}_{l'=1} ,\{\hat{g}^{l,\rm c}_{b,k}\}^{K,B}_{k=1,b=1}$~as the set of known information. Hence, by utilizing the lower bound for a fading channel with additive non-Gaussian noise and side information provided in \cite[Sec.~2.3]{redbook}, we have the lower bound on the capacity of the D2D communication of pair~$l$~as seen in \eqref{eq:bound dd dedicated} at the top of the next page
\begin{figure*}
{{\begin{equation}
		\begin{aligned}\label{eq:bound dd dedicated}
			&R^{\rm d}_{l} = \left(1 - \frac{\tau }{\tau_c}\right)\times \\
			&{\mathbb E}\left\{ \log_2 \left(1+\overbrace{\frac{p^{\rm d}_{l} \left|{\mathbb E}\left[(\hat{g}^{l,\rm d}_{l})^{*} g^{l,\rm d}_{l}\lvert\Omega\right]\right|^2}{ p^{\rm d}_{l} \var\left\{(\hat{g}^{l,\rm d}_{l})^{*} g^{l,\rm d}_{l}\lvert\Omega\right\} + \var\left\{\sum\limits_{b=1}^{B}\sum\limits_{k=1}^{K}\sqrt{p^{l,\rm c}_{b,k}}(\hat{g}^{l,\rm d}_{l})^{*} g^{l,\rm c}_{b,k} s^{\rm c}_{b,k} + \sum\limits_{l'=1, l'\neq l}^{L}\sqrt{p^{\rm d}_{l'}} (\hat{g}^{l,\rm d}_{l})^{*}{g}^{l,\rm d}_{l'} s^{\rm d}_{l'} + (\hat{g}^{l,\rm d}_{l})^{*} w_{l} \lvert\Omega\right\}}}^{{{\rm  \overline{SINR}}^{\rm d}_{l}}}\right)\right\}\\
			&=\left(1 - \frac{\tau }{\tau_c}\right){\mathbb E}\left\{ \log_2 \left(1+\frac{p^{\rm d}_l \left|\hat{g}^{l,\rm d}_{l}\right|^2}{ p^{\rm d}_l \left({\beta}^{l,\rm d}_{l} - {\gamma}^{l,\rm d}_{l} \right) +  \sum\limits_{b=1}^{B}\sum\limits_{k=1}^{K} p^{\rm c}_{b,k} \left(\left|\hat{g}^{l,\rm c}_{b,k}\right|^2 + {\beta}^{l,\rm c}_{b,k} - {\gamma}^{l,\rm c}_{b,k}\right) + \sum\limits_{l'=1,l'\neq l }^{L} p^{\rm d}_{l'} \left(\left|\hat{g}^{l,\rm d}_{l'}\right|^2 + {\beta}^{l,\rm d}_{l'} - {\gamma}^{l,\rm d}_{l'}\right)+1 }\right)\right\}.
		\end{aligned}
		\end{equation}}}
	\hrulefill
\end{figure*}
where the numerator of the effective SINR, denoted as ${\rm  \overline{SINR}}^{\rm d}_{l}$ in \eqref{eq:bound dd dedicated}, is the power of the desired signal. The first term in the denominator is the variance of channel estimation error from the desired D2D transmitter, caused by the gain uncertainty in the detector. The second term is interference from cellular users and the third term includes the interference from other D2D transmitters. This expression is useful for SE computation but it is not in closed-form as it contains an expectation with respect to the small-scale fading coefficients, thus it is not suitable for power control optimization. Hence, to provide a tractable power control algorithm, we compute an approximation of~\eqref{eq:bound dd dedicated}~by computing the expectation of the numerator and the denominator of the fraction inside the logarithm. The resulting approximation of the SE of D2D pair~$l$~is denoted as~$\tilde{R}^{\rm d}_{l}$~and is given by
\begin{equation}\label{eq:app bound dd dedicated final}
\begin{aligned}
		&R^{\rm d}_{l} \approx \tilde{R}^{\rm d}_{l} = \left(1 - \frac{\tau }{\tau_{c}}\right) \log_2 \left(1+ {\rm  SINR}^{\rm d}_{l}\right),
	\end{aligned}
\end{equation}
	where ${\rm  SINR}^{\rm d}_{l}$ is seen in \eqref{eq:app bound dd dedicated final pilot} at the bottom of the next page.
\begin{figure*}[b!]
	\hrulefill
\begin{equation}\label{eq:app bound dd dedicated final pilot}
\begin{aligned} 
	{\rm  SINR}^{\rm d}_{l} &= \frac{{\mathbb E} \left\{p^{\rm d}_{l} \left|\hat{g}^{l,\rm d}_{l}\right|^2\right\}}{ {\mathbb E} \left\{p^{\rm d}_l \left({\beta}^{l,\rm d}_{l} - {\gamma}^{l,\rm d}_{l} \right) +  \sum\limits_{b=1}^{B}\sum\limits_{k=1}^{K} p^{\rm c}_k \left(\left|\hat{g}^{l,\rm c}_{b,k}\right|^2 + {\beta}^{l,\rm c}_{b,k} - {\gamma}^{l,\rm c}_{b,k}\right) + \sum\limits_{l'=1,l'\neq l }^{L} p^{\rm d}_{l'} \left(\left|\hat{g}^{l,\rm d}_{l'}\right|^2 + {\beta}^{l,\rm d}_{l'} - {\gamma}^{l,\rm d}_{l'}\right)+1 \right\}} \\
	& = \frac{ p^{\rm d}_l {\gamma}^{l,\rm d}_{l}}{p^{\rm d}_l \left({\beta}^{l,\rm d}_{l} - {\gamma}^{l,\rm d}_{l} \right) +  \sum\limits_{b=1}^{B}\sum\limits_{k=1}^{K} p^{\rm c}_{b,k} {\beta}^{l,\rm c}_{b,k} + \sum\limits_{l'=1,l'\neq l }^{L} p^{\rm d}_{l'} {\beta}^{l,\rm d}_{l'} +1 } \\
	& = \frac{ \tau p^{\rm d}_l {p}^{\rm p,d}_{l}\left(\beta^{l,\rm d}_{l}\right)^{2} }{ \left({1+ \sum\limits_{l''\in n_{i}}\tau p^{\rm p,d}_{l''} \beta^{l,\rm d}_{l''}}\right)\left(1 +  \sum\limits_{b=1}^{B}\sum\limits_{k=1}^{K} p^{\rm c}_{b,k} {\beta}^{l\rm c}_{b,k} + \sum\limits_{l'=1,l'\neq l }^{L} p^{\rm d}_{l'} {\beta}^{l,\rm d}_{l'}\right)+ \left({p}^{\rm d}_{l}\beta^{l,\rm d}_{l} + {p}^{\rm d}_{l}\beta^{l,\rm d}_{l} \sum\limits_{l''\in n_{i}\backslash l}\tau p^{\rm p,d}_{l''} \beta^{l,\rm d}_{l''} \right) }, 
	\end{aligned}
\end{equation}
\end{figure*}
Note that in \eqref{eq:app bound dd dedicated final pilot} the last expression follows from replacing the mean square of the channel estimates with their actual representation. The tightness of this approximation is investigated in Appendix \ref{appendix:Approximate-SE}.

\section{Optimization of Power Allocation}
\label{optimization}
In this section, which contains the main contribution of the paper, we investigate different power control schemes for the studied system model. We consider power control where the data transmission powers of CUs and D2D transmitters are selected to maximize the max-min fairness and max product SINR optimization objectives. We also propose a power control scheme for joint pilot and data transmission optimization powers for these two objectives, which is considerable more complicated. In all these cases, we consider both the cases of MR and ZF processing at the BSs, which require different solution algorithms. Note that in all optimization problems provided in this section,~$P_{\rm max}>0$~denotes the maximum transmit power of the users. In {Table} \ref{table:1}, we summarize the different problems that we solve and the corresponding subsections where the solutions are found.
\begin{table}[hpbt!]
	\vspace{0.2in}
	\centering
	\caption{Optimization problems and the sections where solution approaches are given.}
	\vspace{-0.1in}
	\begin{center}
		\begin{tabular}{| p{3.8cm} | p{1.8cm}| p{1.7cm} | }
			\hline
			 & MR processing & ZF processing\\
			\hline
			max-min, data power control& Section~\ref{max-min data}&Section~\ref{max-min data} \\
			\hline 
			max-min, data \& pilot power control& Section~\ref{MRC Max-min fairness pilot}&Section~\ref{ZF Proportional fairness pilot} \\
			\hline
			max product SINR, data power control&Section~\ref{max-prod data} &Section~\ref{max-prod data} \\
			\hline
			max product SINR, data \& pilot power control & Section~\ref{MRC Proportional fairness pilot}& Section~\ref{ZF Proportional fairness pilot}\\
			\hline
		\end{tabular}\label{table:1}
	\end{center} 
\end{table}

\subsection{Data power control}
\label{Data power control}
In this subsection, we select the data powers in the system to either maximize the minimum SE in the network (i.e., provide a uniformly high SE among the users) or the product of the SINRs (i.e., the sum of the logarithms of the SINRs, which is a variation of classical sum-SE optimization that guarantees non-zero SE to all users \cite[Sec.~7.1]{bjornson2017massive}).

\subsubsection{Max-min fairness}
\label{max-min data}
We will now select the data powers in the system to maximize the minimum SE of the CUs in all cells as well as all the D2D pairs, for both the cases of MR and ZF processing at the BSs. The max-min problem is formulated in epigraph form as 
\begin{maxi!}[2]
{\substack{\{p^{\rm d}_l,p^{\rm c}_{b,k}\},\lambda}}{\lambda \label{eq:opti-objective_data}}
{\label{eq:opti-original-epigraph_data}}{}
\addConstraint{{R^{b,\rm c}_k}\geq \lambda ~~\forall b,k,\quad{\tilde{R}^{\rm d}_l} \geq \lambda~~\forall l \label{eq:C1optidata}}{}{}
\addConstraint{ 0\leq p^{\rm d}_l}{\leq P_{\rm max}~~\forall l \label{eq:C2opti_data}}
\addConstraint{ 0\leq p^{\rm c}_{b,k}}{\leq P_{\rm max}~~ \forall b,k, \label{eq:C3opti_data}}
\end{maxi!}
where~$\lambda$~indicates the SE level that is guaranteed to all CUs and all D2D pairs and this variable is to be maximized. Note that we use the SE approximation in~\eqref{eq:app bound dd dedicated final}~for the D2D pairs to achieve a tractable problem formulation, while the exact SE will be used later to evaluate the performance. To solve~\eqref{eq:opti-original-epigraph_data}, we fix~$\lambda$~and solve the resulting linear feasibility optimization problem which determines if the given SE level $\lambda$ is achievable or not. 
This is a linear program that can be solved efficiently using CVX~\cite{cvx}, or some other standard solver for linear programs.
To find the optimal~$\lambda$, we can perform a line search over the interval~$ [0,\lambda^{\rm u}]$, where~$\lambda^{\rm u}$~is an upper bound on the SE level that can be guaranteed to all users.
One way to obtain $\lambda^{\rm u}$ is the compute upper bounds on the SEs achievable by the cellular and D2D users by neglect the interference terms and assuming that everyone transmits at maximum power. The user that achieves the lowest SE under these circumstances limits the SE level that can be guaranteed to all the users.  Mathematically, this can be computed as
\begin{equation}\label{eq:utopia_mrc}
\begin{split}
\begin{aligned}	
\lambda^{\rm u}= &	
\underset{ l, k, b,}\min\Big\{\log_2\left(1+P_{\max}p^{\rm d}_l {\gamma}^{\rm d}_{l,l}\right),{\log_2\left(1+P_{\max}(M)\gamma^{b,\rm c}_{b,k}\right)}\Big\} 
\end{aligned}
\end{split}
\end{equation}
in the case of MR processing. With ZF processing, it can be computed as
\begin{equation}\label{eq:utopia_zf}
\begin{split}
\begin{aligned}	
\lambda^{\rm u}= & \underset{l, k, b}\min \Big\{\log_2\left(1+P_{\max}p^{\rm d}_l {\gamma}^{\rm d}_{l,l}\right),\\&{\log_2\left(1+P_{\max}(M-({K+N}))\gamma^{b,\rm c}_{b,k}\right)}\Big\}.
\end{aligned}
\end{split}
\end{equation}
We use the bisection line search algorithm for solving problem~\eqref{eq:opti-original-epigraph_data} \cite{Boyd}, as described in {{Algorithm}}~$1$. This algorithm can solve problem \eqref{eq:opti-original-epigraph_data} for both case of MR and ZF processing at the BS, although the optimal solution will be different. 

{Algorithm}~$1$ consists of solving the linear program~\eqref{eq:opti-original-epigraph_data} that has complexity of $\mathcal{O}(\left(KB+L+1\right)^2 \left(KB+L\right))$, in addition, total number of iteration for convergence of bisection line search is given as $\ceil{\log_2\left(\frac{\lambda^{u}}{\epsilon}\right)}$. Therefore, {{Algorithm}}~$1$ has computational complexity order of $\mathcal{O}(\ceil{\log_2\left(\frac{\lambda^{u}}{\epsilon}\right)}\left(KB+L+1\right)^2 \left(KB+L\right))$ \cite[Ch.~1]{Boyd}. The exact complexity depends on which solver of linear programs that is used and on the channel realizations that appear in the instance of the problem that is being solved.

\begin{algorithm}[t]
	\label{alg}
	\SetAlgoLined
	\textbf{Input:} $\lambda^{\rm l}= 0$ and $\lambda^{\rm u}$,\; line search accuracy $\epsilon$
	
	\While{$\lambda^{\rm u}- \lambda^{\rm l} >\epsilon$}{
		Set $\lambda = \frac{\left(\lambda^{\rm l}+ \lambda^{\rm u}\right)}{2}$\;\\
		Solve the feasibility problem of finding $\{p^{\rm d}_l\}^{L}_{l=1},\{p^{\rm c}_k\}^{K}_{k=1}$ such that \eqref{eq:C1optidata}-\eqref{eq:C1optidata} are satisfied.\\
		\eIf{the feasibility problem is not satisfied}{
			Set $\lambda^{\rm u} = \lambda$\;
		}{
			Set $\lambda^{\rm l} = \lambda $\;\\
			Set $\{p^{\rm d}_l\}^{L}_{l=1},\{p^{\rm c}_k\}^{K}_{k=1} $ as solution for \eqref{eq:opti-original-epigraph}.
		}
	}
	\textbf{Output:}{ Optimal value of $\{p^{\rm d}_l\}^{L}_{l=1},\{p^{\rm c}_k\}^{K}_{k=1}$. }
	\caption{ Bisection Algorithm}
\end{algorithm}

\subsubsection{Max product SINR}
\label{max-prod data}
The max-min fairness power control approach proposed in the previous section may lead to low overall performance. This is due to the fact that the network-wide performance is limited by the channel condition of the most unfortunate users in the network. When increasing the number of active users in the network, the risk of having an extremely bad channel for some of the active users is high. This is handled in~\cite{redbook}~by dropping such users from service before applying the max-min power control. Another potential approach is to perform max product SINR power control instead. To investigate this, we will now change the optimization objective to maximize the product of the SINRs of all users in the network. It offers some level of fairness as well as approximately maximizing the sum SE of the network \cite[Sec.~7.1]{bjornson2017massive}. The optimization problem is formulated as
\begin{maxi!}[2] 
	{\substack{\{p^{\rm d}_l,p^{\rm c}_{b,k}\},\lambda^{b,\rm c}_{b,k}~\zeta^{\rm d}_{l}}}{\prod_{l=1}^{L}~\zeta^{\rm d}_{l}~\prod_{b=1}^{B}\prod_{k=1}^{K_b}\lambda^{b,\rm c}_{b,k}  \label{eq:maxprod-objective1-data}}
	{\label{eq:opti max product1-data}}{}
	\addConstraint{{{\rm SINR}^{b,\rm c}_{b,k}}\geq \lambda^{b,\rm c}_{b,k}  ~~\forall k,b \label{eq:C1maxprod0-data}}{}{}{}
	\addConstraint{{\rm SINR}^{\rm d}_l \geq \zeta^{\rm d}_{l}~~\forall l \label{eq:C1maxprod1-data}}{}{}
	\addConstraint{ 0\leq p^{\rm d}_l}{\leq P_{\rm max},~~\forall l \label{eq:C2paxprod1-data}}
	\addConstraint{ 0\leq p^{\rm c}_{b,k}}{\leq P_{\rm max},~~ \forall k,b \label{eq:C3maxprod1-data}}
\end{maxi!}
where for the case of MR processing the~${\rm SINR}^{b,\rm c}_{b,k}$~in~\eqref{eq:C1maxprod0-data}~constraint is given in~\eqref{eq:MRC-pilot power}~and for ZF processing we replace it with the SINR expression given in~\eqref{eq:ZF-pilot power}. The SINR constraints can be rearranged as geometric constraints, thus the max product SINR with data power control is a geometric program and it can be solved efficiently by using standard convex optimization solvers such as CVX~\cite{cvx}. The interested reader can find the basic terminology and property of geometric programming in Appendix \ref{appendix:Useful definitions}.

\subsection{Joint pilot and data power control for MR processing}
In the previous part, it was assumed that the pilot transmit power of the CUs and D2D transmitters are fixed, e.g., at its maximum value. This is the simplest assumption for pilot transmission and makes perfect sense in a system without pilot contamination. However, in the studied multi-cell scenario with pilot contamination, we can improve the performance by optimizing the pilot transmission powers in addition to data power control. Hence, the goal of this subsection is to perform joint pilot and data power control for the case of MR processing at the BSs with the same optimization objectives as above.
\subsubsection{Max-min fairness}
\label{MRC Max-min fairness pilot}
In this part, similar to Section~\ref{max-min data}, the optimization objective is to maximize the minimum SE of the CUs in all cells as well as all the D2D pairs. The difference from \eqref{eq:opti-original-epigraph_data} is that we have extra variables and constraints when we want to optimize the pilot powers as well. The optimization problem is formulated in epigraph form as
\begin{maxi!}[2]
	{\substack{\{p^{\rm d}_l,p^{\rm c}_{b,k},p^{\rm p, c}_{b,k},p^{\rm p, d}_l\},\lambda}}{\lambda \label{eq:opti-objective}}
	{\label{eq:opti-original-epigraph}}{}
	\addConstraint{R^{b,\rm c}_k \left({\rm MR}\right) \geq \lambda ~~\forall k,b,\quad\tilde{R}^{\rm d}_l \geq \lambda~~\forall l \label{eq:C1,C2opti}}{}{}
	\addConstraint{ 0\leq p^{\rm d}_l}{\leq P_{\rm max},~~\forall l \label{eq:C3p43}}
	\addConstraint{ 0\leq p^{\rm c}_{b,k}}{\leq P_{\rm max},~~ \forall k,b \label{eq:C4p43}}
	\addConstraint{ 0\leq p^{\rm p, d}_l}{\leq P_{\rm max},~~\forall l  \label{eq:C5p43}}
	\addConstraint{ 0\leq p^{\rm p, c}_{b,k}}{\leq P_{\rm max},~~\forall k,b \label{eq:C6p43}},
\end{maxi!}
where~$\lambda$~indicates the SE that is guaranteed to all CUs and all D2D pairs and this variable is to be maximized. By rearranging in the SINR constraints, this optimization problem is in the form of geometric program and as it is mentioned in previous part, this type of optimization problem can be solved efficiently by standard convex optimization solvers such as CVX~\cite{cvx}.
\subsubsection{Max product SINR}
\label{MRC Proportional fairness pilot}
For the case of MR processing at the BSs, the optimization problem is formulated as
\begin{maxi!}[2]
	{\substack{\{p^{\rm d}_l,p^{\rm c}_{b,k},p^{\rm p, d}_l,p^{\rm p, c}_{b,k}\},\lambda^{b,\rm c}_{b,k}~\zeta^{\rm d}_{l}}}{\prod_{l=1}^{L}~\zeta^{\rm d}_{l}~\prod_{b=1}^{B}\prod_{k=1}^{K_b}\lambda^{b,\rm c}_{b,k}  \label{eq:maxprod-objective1}}
	{\label{eq:opti max product1}}{}
	\addConstraint{{{\rm SINR}^{b,\rm c}_{b,k} \left({\rm MR}\right)}\geq \lambda^{b,\rm c}_{b,k}  ~~\forall k,b}{}{}{}
	\addConstraint{{\rm SINR}^{\rm d}_l \geq \zeta^{\rm d}_{l}~~\forall l \label{eq:C1maxprod1}}{}{}
	\addConstraint{ 0\leq p^{\rm d}_l}{\leq P_{\rm max},~~\forall l \label{eq:C2paxprod1}}
	\addConstraint{ 0\leq p^{\rm c}_{b,k}}{\leq P_{\rm max},~~ \forall k,b \label{eq:C3maxprod1}}
	\addConstraint{ 0\leq p^{\rm p,d}_l}{\leq P_{\rm max},~~\forall l \label{eq:C4maxprod1}}
	\addConstraint{ 0\leq p^{\rm p, c}_{b,k}}{\leq P_{\rm max},~~\forall b,k \label{eq:C5maxprod1}},
	\vspace{-.3cm}
\end{maxi!}
where in the first and second constraints, the SINR expressions are provided from the SE expression given in \eqref{eq:MRC-pilot power} and \eqref{eq:app bound dd dedicated final pilot}, respectively. Similar to \eqref{eq:opti max product1-data}, the SINR constraints can be rearranged as geometric constraints, thus \eqref{eq:opti max product1} can be solved efficiently as a geometric program.

\subsection{Joint pilot and data power control for ZF processing}
\label{ZF Proportional fairness pilot}
In this subsection, we propose an efficient algorithm for pilot and data power control for the case of ZF processing at the BSs. Due to the more complicated structure of the SINRs with ZF, we can only solve the power control problems to local optimality.\footnote{Global optimization techniques might be used to find the global optimum, but these algorithms have a prohibitive computational complexity for most practical purposes.} We focus on max product SINR optimization problem, but it is straightforward to apply the same methodology for the case of max-min fairness optimization problem. The optimization problem of interest is formulated as
\begin{maxi!}[2]
	{\substack{\{p^{\rm d}_l, p^{\rm c}_{b,k}, p^{\rm p,d}_l, p^{\rm p, c}_{b,k}\},\lambda^{b,\rm c}_{b,k}~\zeta^{\rm d}_{l}}}{\prod_{l=1}^{L}~\zeta^{\rm d}_{l}~\prod_{b=1}^{B}\prod_{k=1}^{K_b}\lambda^{b,\rm c}_{b,k}  \label{eq:ZF-maxprod-objective}}
	{\label{eq:opti max product}}{}
	\addConstraint{{{\rm SINR}^{b,\rm c}_{b,k}\left({\rm ZF}\right)}\geq \lambda^{b,\rm c}_{b,k}  ~~\forall k,b \label{eq:C1maxprod}}{}{}{}
	\addConstraint{{\rm SINR}^{\rm d}_l \geq \zeta^{\rm d}_{l}~~\forall l \label{eq:C1-1maxprod}}{}{}
	\addConstraint{ 0\leq p^{\rm d}_l}{\leq P_{\rm max},~~\forall l \label{eq:C2paxprod}}
	\addConstraint{ 0\leq p^{\rm c}_{b,k}}{\leq P_{\rm max},~~ \forall k,b \label{eq:C3maxprod}}
	\addConstraint{ 0\leq p^{\rm p,d}_l}{\leq P_{\rm max},~~\forall l}
	\addConstraint{ 0\leq p^{\rm p, c}_{b,k}}{\leq P_{\rm max},~~\forall b,k},
\end{maxi!} 
where the SINR constraint in~\eqref{eq:C1maxprod}~is provided in the SE expression given in~\eqref{eq:ZF-pilot power}. The optimization constraints~\eqref{eq:C1maxprod}~are not posynomial; hence, the optimization problem is not a geometric program (see Appendix~\ref{appendix:Useful definitions} for details). In fact, the SINR constraints are signomial constraints, thus there is no way to reformulate them to make the problem efficiently solvable, with resorting to approximations. The signomial constraints are in the form of fractions of two posynomial functions. To address this problem, we develop an algorithm that finds a local optimum. We propose to first lower bound the posynomials function in the denominator of the signomial constraint with a monomial function.  Then, the resulting function is a fraction of a posynomial with the monomial which is also a posynomial function. Consequently, the lower-bounded SINR constraints is a geometric constraint, so that we have obtained an approximation of the original problem that can be solved as a geometric program, using standard solvers such as CVX~\cite{cvx}. By successively updating the approximation, a local optimum is guaranteed to be found. The detailed procedure is described below.
\begin{lemma}\label{lemma1}
\cite[Lemma 1]{gpPowercont} Suppose $u_j\left({\mathbf x}\right)$ is a monomial function. A posynomial function $f\left({\mathbf x}\right)$ is given as
	\begin{equation}
		f\left({\mathbf x}\right) = \sum_{j} u_j\left({\mathbf x}\right)
	\end{equation}
		which can be lower bounded by the following monomial function $\tilde{f}\left({\mathbf x}\right)$ 
		\begin{equation}
		 f\left({\mathbf x}\right) \geq \tilde{f}\left({\mathbf x}\right) = \prod_{j} \left(\frac{u_j\left({\mathbf x}\right)}{Q_j}\right)^{Q_j},
		\end{equation}		
\end{lemma}
where~${Q_j}\ge 0$~is a weight. For any fixed and element-wise positive ${\mathbf x}_0$, $\tilde{f}\left({\mathbf x}\right)$ is the best local monomial approximation of $f\left({\mathbf x}\right)$ near ${\mathbf x}_0$ in the sense of first order Taylor approximation if we choose the weight as
\begin{equation}\label{eq:approx-weight}
Q_j = \frac{u_j\left({\mathbf x}_0\right)}{f\left({\mathbf x}_0\right)}.
\end{equation}

\begin{proof}
	The proof is provided in \cite[Lemma 1]{gpPowercont}.
	\end{proof}
In the SINR expression of ZF processing in~\eqref{eq:ZF-pilot power},~we have~\eqref{eq:ZF-pilot power-denominator}~in the denominator. In~\eqref{eq:ZF-pilot power-denominator}, we find 
\begin{equation}
\begin{aligned} 
f\left(p^{\rm p,c}_{b'',k''}\right) &= \left(\frac{1+\tau \sum\limits_{b'\neq b''}^{B} p^{\rm p,c}_{b',k''} \beta^{b,\rm c}_{b',k''}}{1+\tau \sum\limits_{b'=1}^{B} p^{\rm p,c}_{b',k''} \beta^{b,\rm c}_{b',k''}}\right),~~\forall k'',b'', \\
\end{aligned}
\end{equation}

\begin{equation}
\begin{aligned}
 g\left(p^{\rm p,d}_{l}\right) &= \left(\frac{1+\tau \sum\limits_{l'\in \mathcal{N}_i,l'\neq l} p^{\rm p,d}_{l'} \beta^{b,\rm d}_{l'}}{1+\tau \sum\limits_{l'\in \mathcal{N}_i} p^{\rm p,d}_{l'} \beta^{b,\rm d}_{l'}}\right),~~\forall l,
 \end{aligned}
\end{equation}
which prevent the denominator from being a posynomial function. However, the denominator of $f\left(p^{\rm p,c}_{b'',k''}\right)$ and $g\left(p^{\rm p,d}_{l}\right)$,  can be lower bounded using Lemma \ref{lemma1} and we have the following upper bound for~$f\left(p^{\rm p,c}_{b'',k''}\right)$ and $g\left(p^{\rm p,d}_{l}\right)$
\begin{equation}
\begin{aligned}
f\left(p^{\rm p,c}_{b'',k''}\right) & \leq \tilde{f}\left(p^{\rm p,c}_{b'',k''}\right)\\
&= \left(\frac{1+\tau \sum\limits_{b'\neq b''}^{B} p^{\rm p,c}_{b',k''} \beta^{b,\rm c}_{b',k''}}{ C\prod\limits_{b'=1}^{B} \left(\frac{\tau \beta^{b,\rm c}_{b',k''} p^{\rm p,c}_{b',k''} }{Q_{b',k''}}\right)^{Q_{b',k''}}}\right),~~\forall, k'',b'',
 \end{aligned}
\end{equation}
\begin{equation}
\begin{aligned}
g\left(p^{\rm p,d}_{l}\right)& \leq \tilde{g}\left(p^{\rm p,d}_{l}\right)= \left(\frac{1+\tau \sum\limits_{l'\in \mathcal{N}_i,l'\neq l} p^{\rm p,d}_{l'} \beta^{b,\rm d}_{l'}}{ D\prod\limits_{l'\in \mathcal{N}_i} \left(\frac{\tau \beta^{b,\rm d}_{l'} p^{\rm p,d}_{l'}}{Q_{l'}}\right)^{Q_{l'}}}\right),~~\forall l,
 \end{aligned}
 \end{equation}
 where  $Q_{b',k''}$ and $ Q_{l'}$ are the weights  calculated by using~\eqref{eq:approx-weight}. In addition, we use \eqref{eq:approx-weight} to calculate~$C$~and~$D$ as
\begin{equation}
 \begin{aligned}
 	C &=  \bigg({1+\tau \sum\limits_{b'=1}^{B} p^{\rm p,c}_{b',k''} \beta^{b,\rm c}_{b',k''}}\bigg)^{\bigg(\frac{1}{1+\tau \sum\limits_{b'=1}^{B} p^{\rm p,c}_{b',k''} \beta^{b,\rm c}_{b',k''}}\bigg)}, \\
\end{aligned}
\end{equation}
\begin{equation}
\begin{aligned}
 	D &= {\bigg(1+\tau \sum\limits_{l'\in \mathcal{N}_i} p^{\rm p,d}_{l'} \beta^{b,\rm d}_{l'}\bigg)}^{\bigg(\frac{1}{1+\tau \sum\limits_{l'\in \mathcal{N}_i} p^{\rm p,d}_{l'} \beta^{b,\rm d}_{l'}}\bigg)},\\
 \end{aligned}
 \end{equation}
that are the corresponding weights for the $1$ in the denominator of $f\left(p^{\rm p,c}_{b'',k''}\right)$ and $g\left(p^{\rm p,d}_{l}\right)$, respectively.

Hence, by using these posynomial upper bounds on ~$f\left(p^{\rm p,c}_{b'',k''}\right)$ and $g\left(p^{\rm p,d}_{l}\right)$, we have the following lower bound on the SINR constraint in~\eqref{eq:C1maxprod}:
\begin{equation}\label{eq:appxSINR}
\begin{aligned}
\widetilde{\rm SINR}^{b,\rm c}_{b,k} & =  \frac{p^{\rm c}_{b,k} \left({\tau p^{\rm p,c}_{b,k} \left({\beta^{b,\rm c}_{b,k}}\right)^2}\right) \left(M-\left(K+N\right)\right) }{\tilde{I}^{b,\rm c}_{b,k} \left(\rm ZF\right)},\\
\end{aligned}
\end{equation}
where $\tilde{I}^{b,\rm c}_{b,k} \left(\rm ZF\right)$ is defined as seen in \eqref{eq:I_ZF} at the top of the page.
\begin{figure*}
\begin{equation}\label{eq:I_ZF}
\begin{aligned}
& \tilde{I}^{b,\rm c}_{b,k} \left(\rm ZF\right) =1+\tau \sum\limits_{b'=1}^{B} p^{\rm p,c}_{b',k} \beta^{b,\rm c}_{b',k} +  \left(1+\tau \sum\limits_{b'=1}^{B} p^{\rm p,c}_{b',k} \beta^{b,\rm c}_{b',k}\right)\sum\limits_{b''=1}^{B}\sum\limits_{k''=1}^{K} \left(\frac{1+\tau \sum\limits_{b'\neq b''}^{B} p^{\rm p,c}_{b',k''} \beta^{b,\rm c}_{b',k''}}{ C\prod\limits_{b'=1}^{B} \left(\frac{\tau \beta^{\rm p,c}_{b',k''} p^{\rm p,c}_{b',k''}}{Q_{b',k''}}\right)^{Q_{b',k''}}}\right) p^{\rm c}_{b'',k''}\beta^{b,\rm c}_{b'',k''} +\\
&+\left(1+\tau \sum\limits_{b'=1}^{B} p^{\rm p,c}_{b',k} \beta^{b,\rm c}_{b',k}\right) \sum\limits_{l=1}^{L}  p^{b,\rm d}_{l}\beta^{b,\rm d}_{l} \left(\frac{1+\tau \sum\limits_{l'\in \mathcal{N}_i,l'\neq l} p^{\rm p,d}_{l'} \beta^{b,\rm d}_{l'}}{ D\prod\limits_{l'\in \mathcal{N}_i} \left(\frac{\tau \beta^{\rm p,d}_{l'} p^{\rm p,d}_{l'}}{Q_{l'}}\right)^{Q_{l'}}}\right)	+\left(M-\left(K+N\right)\right) \sum\limits_{b''=1, b'' \neq b} \tau  p^{\rm c}_{b'',k} p^{\rm p,c}_{b'',k} \left({\beta^{b,\rm c}_{b'',k}}\right)^2.
\end{aligned}
\end{equation} 
\hrulefill
\end{figure*}                             
We utilize the lower bounded SINR expression in~\eqref{eq:appxSINR}~in the optimization problem~\eqref{eq:ZF-maxprod-objective}~and can then solve it using geometric programming. To achieve a locally optimum to the original optimization problem, we propose a successive approximation algorithm in which we iteratively solve the optimization problem with the approximated SINR expression for CUs and update the pilot power coefficients of CUs and the D2D transmitters iteratively until a convergence criteria is satisfied. By using \eqref{eq:approx-weight},  for iteration $i$ we have
\begin{equation}
\begin{aligned}\label{eq:weightupdates}
Q^{\left(i\right)}_{b',k''} &= \frac{\tau p^{\rm p,c}_{b',k''}\left(i-1\right)}{1+\tau \sum\limits_{b'=1}^{B} p^{\rm p,c}_{b',k''}\left(i-1\right) \beta^{b,\rm c}_{b',k''}},\\
Q^{\left(i\right)}_{l'} &= \frac{\tau p^{\rm p,d}_{l'}\left(i-1\right)}{1+\tau \sum\limits_{l' \in \mathcal{N}_i} p^{\rm p,d}_{l'}\left(i-1\right) \beta^{b,\rm d}_{l'}},
\end{aligned}
\end{equation}
\begin{equation}
\begin{aligned}\label{eq:weight_one_updates}
C^{\left(i\right)} = \left({1}{\Big/}{C^{\left(i\right)}_1}\right)^{C^{\left(i\right)}_1},\quad\quad
 D^{\left(i\right)} = \left({1}{\Big/}{D^{\left(i\right)}_1}\right)^{D^{\left(i\right)}_1},
\end{aligned}
\end{equation}
where~$C^{\left(i\right)}_1 = \frac{1}{1+\tau \sum\limits_{b'=1}^{B} p^{\rm p,c}_{b',k''}\left(i-1\right) \beta^{b,\rm c}_{b',k''}}$~and~$D^{\left(i\right)}_1  = \frac{1}{1+\tau \sum\limits_{l' \in \mathcal{N}_i} p^{\rm p,d}_{l'}\left(i-1\right) \beta^{b,\rm d}_{l'}}$. In addition, $p^{\rm p,c}_{b',k''}\left(i-1\right),p^{\rm p,d}_{l'}\left(i-1\right)$ are the points at arbitrary iteration $i-1$ of the algorithm which the approximations are applied. The detailed steps is provided in {{Algorithm}}~$2$. 
\begin{algorithm}[hbpt!]
	\label{alg2}
	\DontPrintSemicolon
	\SetAlgoLined
	\textbf{Input:} Initialize $p^{\rm p,c}_{b',k''} \left(0\right) ,p^{\rm p,d}_{l'}\left(0\right) $ for $i=0$; $\epsilon' > 0$ \\
	\Repeat{the convergence is satisfied: \begin{align*}
		 &\left|p^{\rm p,c}_{b',k''}\left(i\right)-p^{\rm p,c}_{b',k''}\left(i-1\right)\right|<\epsilon' \text{and}\\
		&\left|p^{\rm p,d}_{l'} \left(i\right)-p^{\rm p,d}_{l'} \left(i-1\right)\right|<\epsilon'\end{align*}\vspace{-5mm}}{
		$i = i+1$\\
		Compute the weight values $Q^{\left(i\right)}_{b',k''}$ ,$Q^{\left(i\right)}_{l'}$, $C^{\left(i\right)}$ and $D^{\left(i\right)}$~by using \eqref{eq:weightupdates} and \eqref{eq:weight_one_updates};\\
		Solve the optimization problem \eqref{eq:opti max product} with the lower bounded SINR constraint given in \eqref{eq:appxSINR}; \\ 
		Put $p^{\rm p,c}_{b',k''} \left(i\right),p^{\rm p,d}_{l'} \left(i\right)$ $p^{\rm c}_{b',k''} \left(i\right)$ and $p^{\rm d}_{l'} \left(i\right)$ equal to solution of the optimization problem;\\
	}
	\textbf{Output:}{Optimal value of $\{p^{\rm p, d}_l, p^{\rm d}_l\}^{L}_{l=1},\{p^{\rm p, c}_k, p^{\rm c}_k\}^{K}_{k=1}$. }
	\caption{Successive Approximation Algorithm}
\end{algorithm}
\begin{lemma}\label{lemma2}
	{{Algorithm}}~$2$ provides a local optimal Karush-Kuhn-Tucker (KKT) point to \eqref{eq:ZF-maxprod-objective}.
\end{lemma}
\begin{proof}
	The lower bound provided for SINR expression of ZF processing at the BSs in \eqref{eq:appxSINR} fulfills the following three conditions: 
	\begin{enumerate}
	\item \begin{equation*}
	\hspace{-.2in}\widetilde{\rm SINR}^{b,\rm c}_{b,k} \left(p^{\rm p, c}_{b'',k''},p^{\rm p, d}_{l}\right) \leq {\rm SINR}^{b,c}_{b,k} \left(p^{\rm p, c}_{b'',k''},p^{\rm p, d}_{l}\right)~\forall k'',b'',l.\end{equation*}
	\item \begin{equation*}
	\begin{aligned}
	&\widetilde{\rm SINR}^{b,\rm c}_{b,k} \left(p^{\rm p, c}_{b'',k''}\left(i-1\right),p^{\rm p, d}_{l}\left(i-1\right)\right)\\ &= {\rm SINR}^{b,c}_{b,k} \left(p^{\rm p, c}_{b'',k''}\left(i-1\right),p^{\rm p, d}_{l}\left(i-1\right)\right) \quad \forall k'',b'',l,
		\end{aligned}\end{equation*} 
	where $p^{\rm p, c}_{b'',k''}\left(i-1\right)$ and $p^{\rm p, d}_{l}\left(i-1\right)$ are the solution at the previous iteration.
	\item
	\begin{equation*}
	\begin{aligned}
	&\nabla\widetilde{\rm SINR}^{b,\rm c}_{b,k} \left(p^{\rm p, c}_{b'',k''}\left(i-1\right),p^{\rm p, d}_{l}\left(i-1\right)\right) \\
	&= \nabla{\rm SINR}^{b,c}_{b,k} \left(p^{\rm p, c}_{b'',k''}\left(i-1\right),p^{\rm p, d}_{l}\left(i-1\right)\right).
	\end{aligned}
	\end{equation*}
\end{enumerate}
The first condition guarantees the feasibility of the solution of the optimization problem~\eqref{eq:opti max product}, when applying the lower bound for the original optimization problem with actual signomial constraint. The second condition ensures the solution from the previous iteration is also feasible for the current iteration while increasing the optimization objective. If the convergence criteria is satisfied the optimal solution of previous iteration is also the solution of current iteration. As the SINR expression is bounded from above and it is a continuous function of $p^{\rm p, c}_{b'',k''}$ and $p^{\rm p, d}_{l}$, it ensures that the iterative algorithm converges to a limit point. In addition, the solution of optimization problem with the lower bounded SINR constraint satisfies the Slater's condition. Hence, the limit point of the successive approximation algorithm after convergence is a local optimal KKT point \cite{marks1978general,van2017joint}.	
\end{proof}
Each iteration of {{Algorithm}}~$2$ consists of solving optimization problem \eqref{eq:opti max product} which has $3(KB+L)$ variables and the same number of constraints. Hence the order of computational complexity of the algorithm is 
	\begin{equation}
	\begin{aligned}
	\mathcal{O}\left(N\max\left\{\left(3L\left(K+1\right)\right)^3, F\right\}\right)
	\end{aligned}
	\end{equation}
	where $N$ is the required number of iteration for convergence of the algorithm and $F$ is the cost of evaluating the first and second derivatives of the objective and constraint functions in \eqref{eq:opti max product} \cite[Ch.~1]{Boyd}. 
Note that the same solution approach can be used for the case of max-min fairness with joint pilot and data power control for ZF processing at the BSs. Therefore, it is omitted to avoid unnecessary repetition in the content.

\section{Numerical Analysis}
In this section, we provide a numerical evaluation of the performance of the power control algorithm proposed in Section~\ref{optimization}. Note that the exact expressions of the SEs are used to compute the SEs. The simulation setup consists of a multi-cell network with $9$ cells as illustrated in Fig.~\ref{fig:model_sim}. 
\begin{figure}[hptb!]
	\centering
	\centerline{\includegraphics[width=0.6\columnwidth]{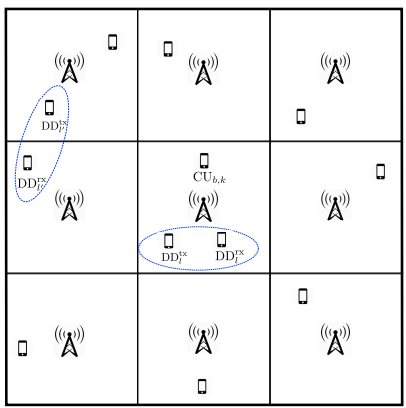}}
	\caption{The simulation setup with a cellular network Massive MIMO system that is underlaid by D2D communications.}
	\label{fig:model_sim}
\end{figure}
In addition, we use a wrap-around topology to avoid edge effects. We assume that the BSs are located in a $1$\,km$^2$ area. Each BS serves 5 CUs which are randomly distributed with uniform distribution in the BS's coverage area. The network contains $10$ D2D pairs randomly distributed in the coverage area with uniform distribution. Since the pairs are distributed randomly, some cells will be more affected by D2D interference than others. We further assume that the distance between the transmitter and receiver of each D2D pair is $10$ meters (note that the effect of D2D distance on spectral efficiency of CUs is further investigated in Appendix \ref{appendix:D2Ddisance}). We use the three-slope pathloss model from \cite{AoTang} to model the large-scale fading in the network. This model contains two reference distances $d_0$ and $d_1$, which are chosen to be $10\,$ and $50\,$ meters, respectively.
The carrier frequency is $2\,$GHz, the bandwidth is $20\,$MHz \cite{3gpp}, and the coherence interval contains $200$ samples. The total number of pilots assigned for the D2D pairs is $N=5$. First, we randomly select 5 D2D pairs and randomly assign a unique pilot sequence to each of them. For the remaining D2D transmitters, we randomly select and reuse a pilot sequence that is already used in the first group. Therefore, for the current setup, the minimum number of D2D pairs that is using a pilot sequence is one and the maximum is six. In addition, the noise variance is set to~$-94\,$dBm\footnote{This is calculated using the noise power spectral density at room temperature multiplied with the $20\,$MHz bandwidth and a
noise figure of $7\,$dB in the receiver hardware.} and the maximum transmit power of the CUs and D2D transmitters are chosen to be~$P_{\rm max}=200\,$mW for both pilot and data transmission.

\subsection{Optimize data power control}
In this subsection, we provide simulation results for data power control while all CUs and D2D transmitters use maximum power for pilot transmission. Figs.~\ref{fig:SE CU}~and~\ref{fig:SE CU ZF}~show the cumulative distribution function (CDF) of the SE of a typical CU $k$ for MR and ZF processing at the BSs, respectively, where the randomness is due to different user locations. We compare the proposed max-min data power control from Section~\ref{Data power control}~and max product SINR power control from Section \ref{max-prod data} with three baseline methods:\footnote{The latter two schemes are obtained as special cases of the proposed schemes.}
\begin{itemize}
\item Equal power control (using full power);
\item  Max-min data power control when there are only CUs;
\item Max product SINR power control when there are only CUs.
\end{itemize}
It can be seen that the performance of our proposed max-min and max product SINR algorithms are almost the same as in the corresponding case without underlaid D2D communication for both MR and ZF. Thus the power control can efficiently mitigate the D2D interference. For the max-min power control, comparing the results with the case of maximum equal power transmission, max-min power control improves the performance of the~$40\%$~weakest users for MR and $0.064\%$ of the weakest users for ZF processing. This is in line with the goal of max-min fairness power control, but it is clear that ZF by itself provides good performance for the most unfortunate CUs.

The proposed max product SINR power control algorithm improves the performance of $80\%$ of users comparing with full power transmission when we have MR processing at the BSs. For the case of max product SINR with ZF processing at the BSs, there are only minor differences in SE in this case. Since ZF can efficiently mitigate the interference, making power control is less important in this case. This shows that with max product SINR, we get almost the same performance as with full power transmission. In fact, ZF mitigates the interference so well that everyone can transmit at full power in the network (Appendix \ref{appendix:power_coeff} provides a further discussion on the effect of ZF processing on the data power coefficients of the D2D transmitters and CUs).
\begin{figure}[b!]
	\centering
	\begin{subfigure}[b]{0.8\columnwidth} \centering 
		\includegraphics[width=1\columnwidth]{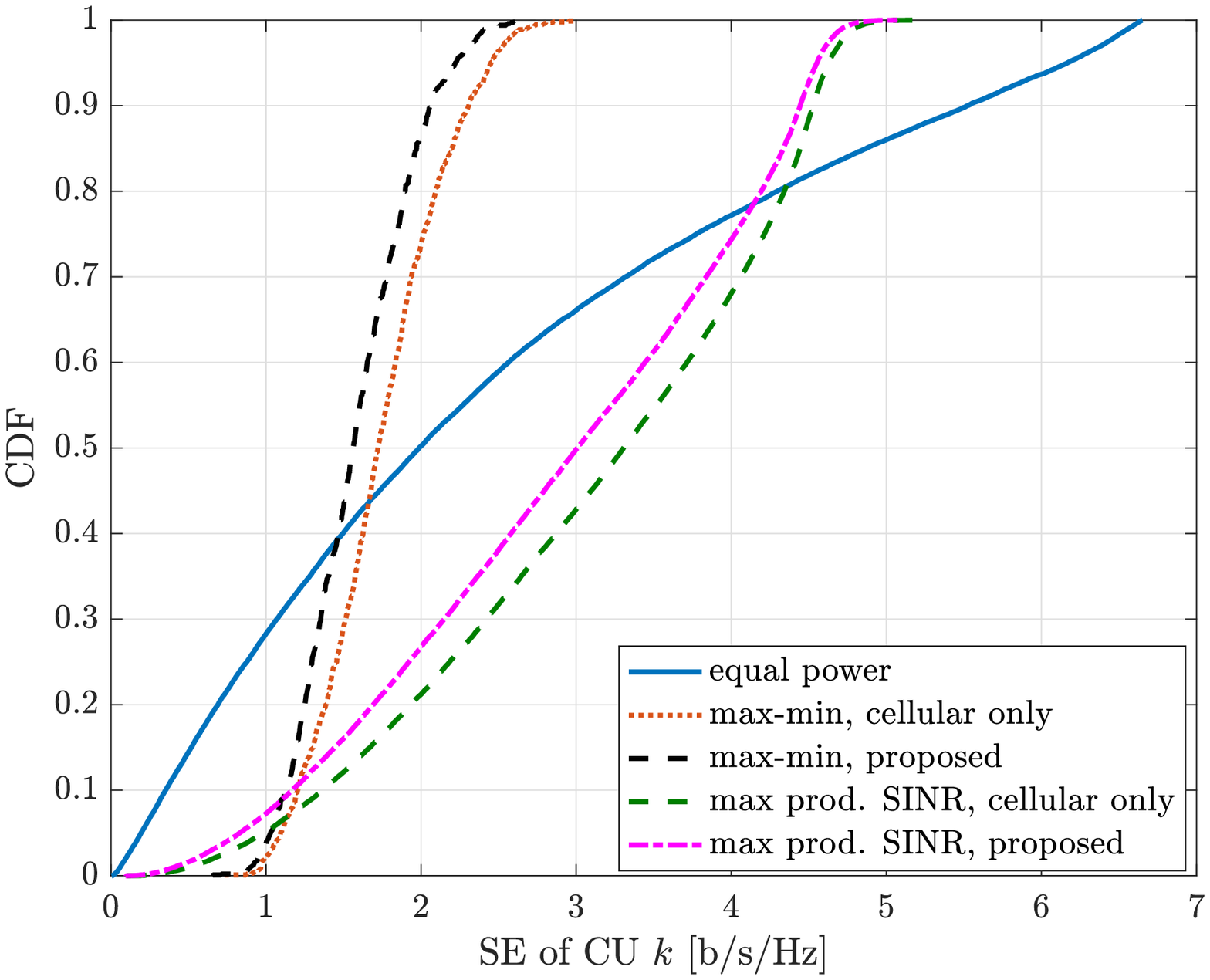} 
		\caption{MR processing}
		\label{fig:SE CU}
	\end{subfigure} 
	\begin{subfigure}[b]{0.8\columnwidth} \centering
		\includegraphics[width=1\columnwidth]{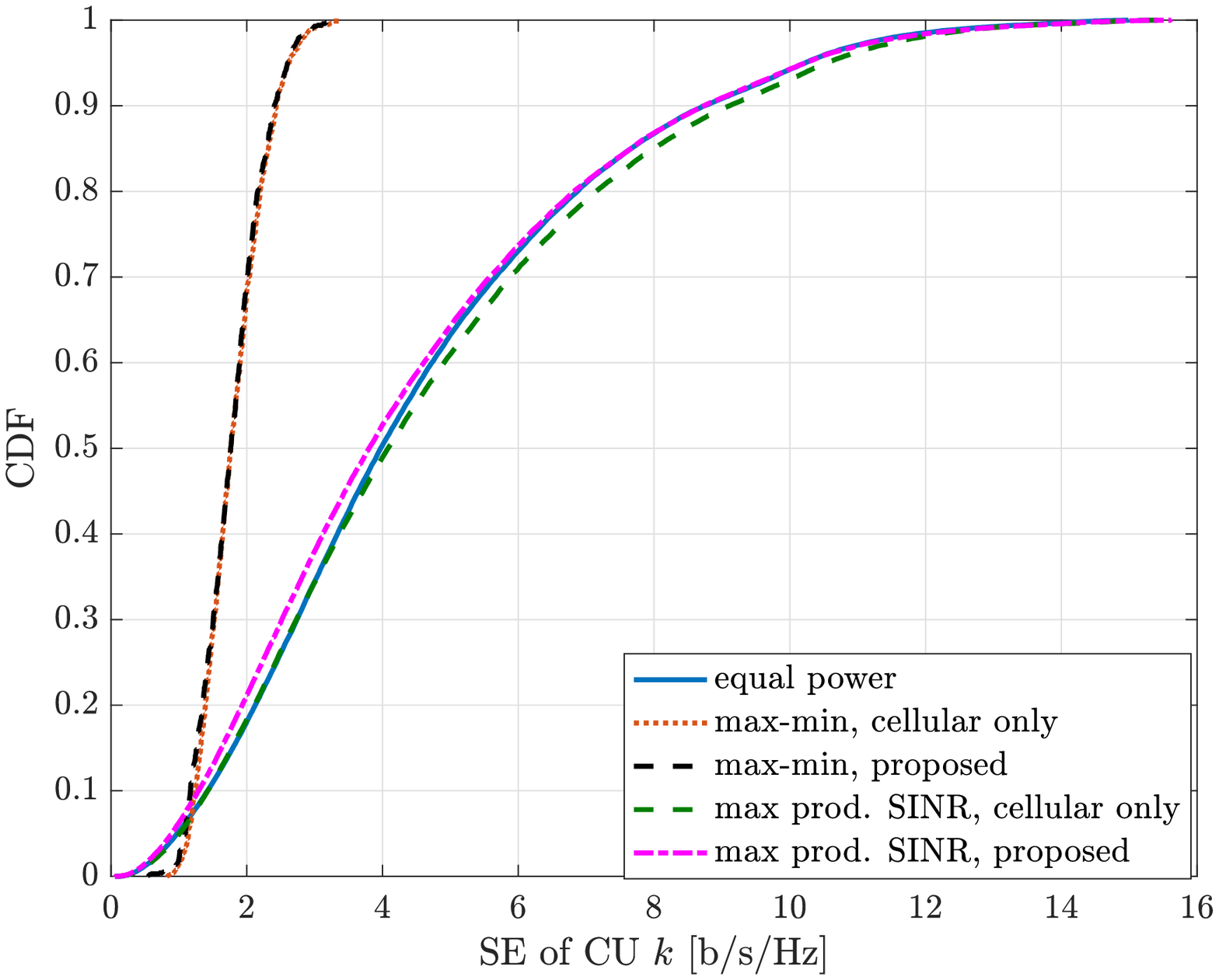} 
		\caption{ ZF processing} 
		\label{fig:SE CU ZF}
	\end{subfigure}
	\caption{SE of arbitrary CU $k$ with data power control.}
	\label{figureDownlinkSimulationK10}  
\end{figure} 
From~Fig.~\ref{fig:SE CU}~and Fig.~\ref{fig:SE CU ZF}~it appears as if the SE of the system reduces when adding D2D users, but that is not the case. Fig.~\ref{fig:Sum SE}~and Fig.~\ref{fig:Sum SE ZF}, show the sum SE of all users in the network when using MR and ZF processing, respectively. It shows higher SE with D2D underlaying. This is due to the fact that we have more active users in the network with D2D underlay. Our proposed max product SINR power control shows a significant improvement in the sum SE for the case of D2D underlaying.
\begin{figure}[hbpt!]
	\centering
	\begin{subfigure}[b]{0.8\columnwidth} \centering 
		\includegraphics[width=\columnwidth]{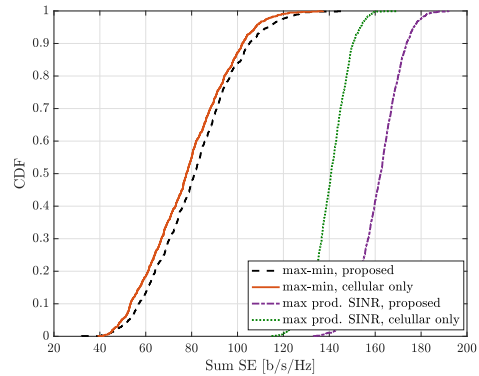} 
		\caption{MR processing}  \vspace{-1mm}
		\label{fig:Sum SE}
	\end{subfigure}  \hspace{3mm}
	\begin{subfigure}[b]{0.8\columnwidth} \centering
		\includegraphics[width=\columnwidth]{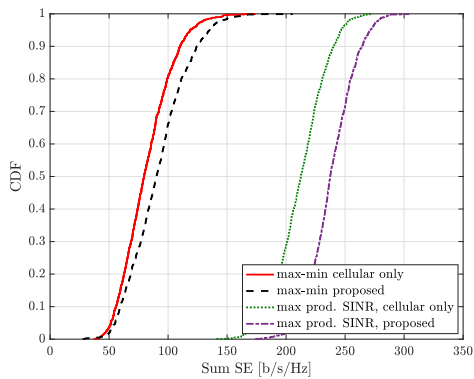} 
		\caption{ ZF processing} \vspace{-3mm}
		\label{fig:Sum SE ZF}
	\end{subfigure}  
	\caption{Sum SE of the network with data power control.}
	\label{fig:Sum SE all}  
\end{figure}

\subsection{Optimized joint pilot and data power control}

Next, we will jointly optimize the pilot and data powers. Fig.~\ref{fig:MRC data versus pilot}  and Fig.~\ref{fig:ZF data versus pilot} show the SE of CU $k$ for both cases of max-min and max product SINR power control with MR and ZF processing at the BSs, respectively. The figures compare the performance with joint pilot and data power control (denoted as ``joint pilot \& data PC'') with the case of only data power control (denoted as ``data PC''). Note that PC stands for power control in the figures. For both MR and ZF processing, the gain from joint pilot and data power control is larger when having the max-min optimization objective. Note that the stopping criterion value $\epsilon'$ in {Algorithm}~$2$ is set to $0.001$ in the simulations.

The sum SE is shown in Fig.~\ref{fig:Sum MRC data versus pilot}~and Fig.~\ref{fig:Sum ZF data versus pilot}, for the cases of max-min and max product SINR power control with MR and ZF processing at the BSs, respectively. In both cases, the figures show that joint pilot and data power control for max-min optimization objective enhances the SE of the network, as compared to only data power control.  When we use max product SINR as the optimization objective, we have clear gains for MR processing which can be seen in Fig.~\ref{fig:Sum MRC data versus pilot}.  But in the ZF processing case in Fig.~\ref{fig:Sum ZF data versus pilot} the curves are almost the same which reveals that ZF processing is capable of mitigating interference efficiently irrespectively of the power control. We observe from Fig.~\ref{fig:data versus pilot} and Fig.~\ref{fig:Sum data versus pilot} that the gain from joint pilot and data power control, as compared to data power control only, is larger when using MR than with ZF. This is due to the fact that ZF can efficiently mitigate in interference spatially, if the channel estimates are good.
\begin{figure}[b !]
	\centering
	\begin{subfigure}[b]{0.8\columnwidth} \centering 
		\includegraphics[width=\columnwidth]{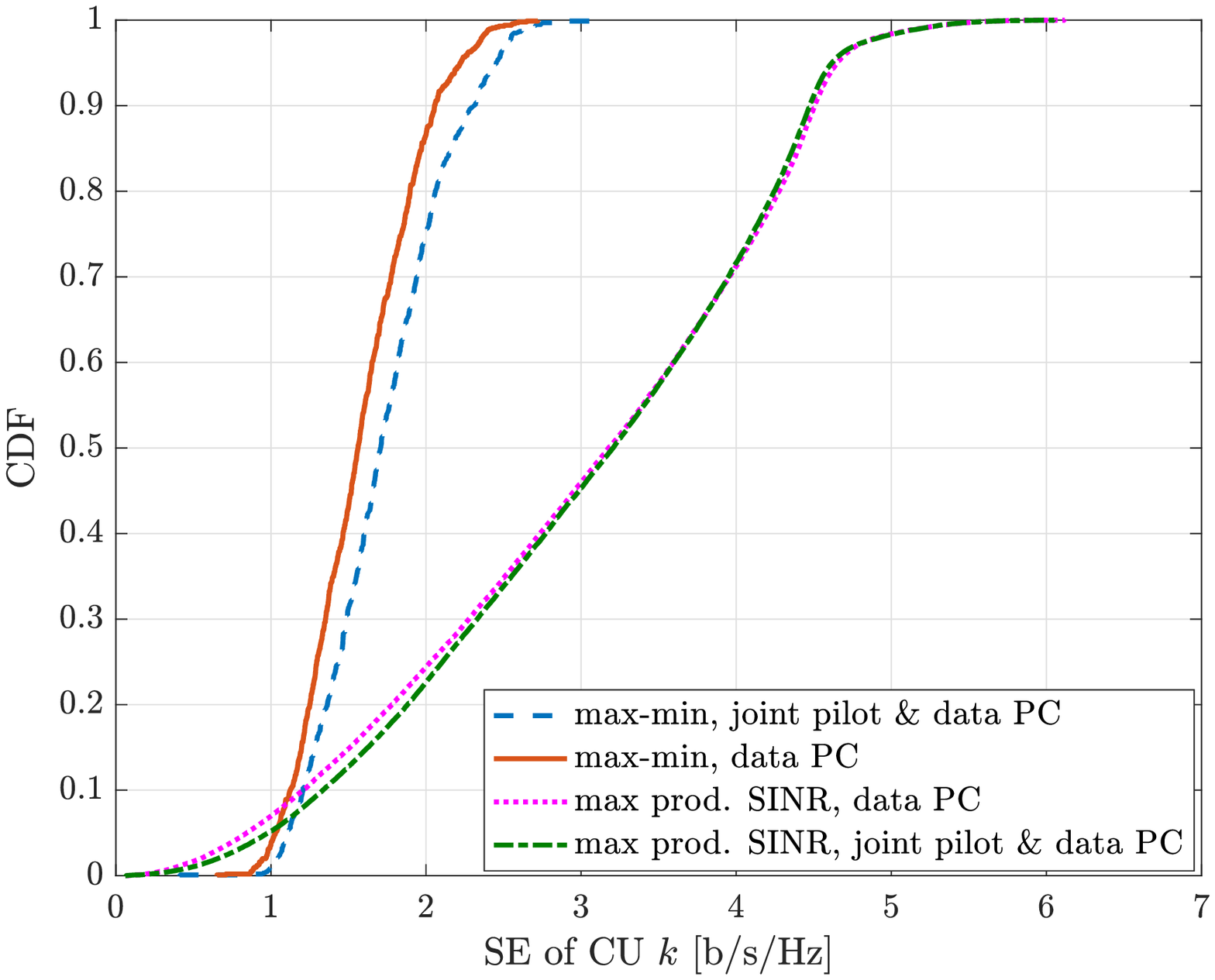} 
		\caption{MR processing}  \vspace{-3mm}
		\label{fig:MRC data versus pilot}
	\end{subfigure}
	\begin{subfigure}[b]{0.8\columnwidth} \centering
		\includegraphics[width=\columnwidth]{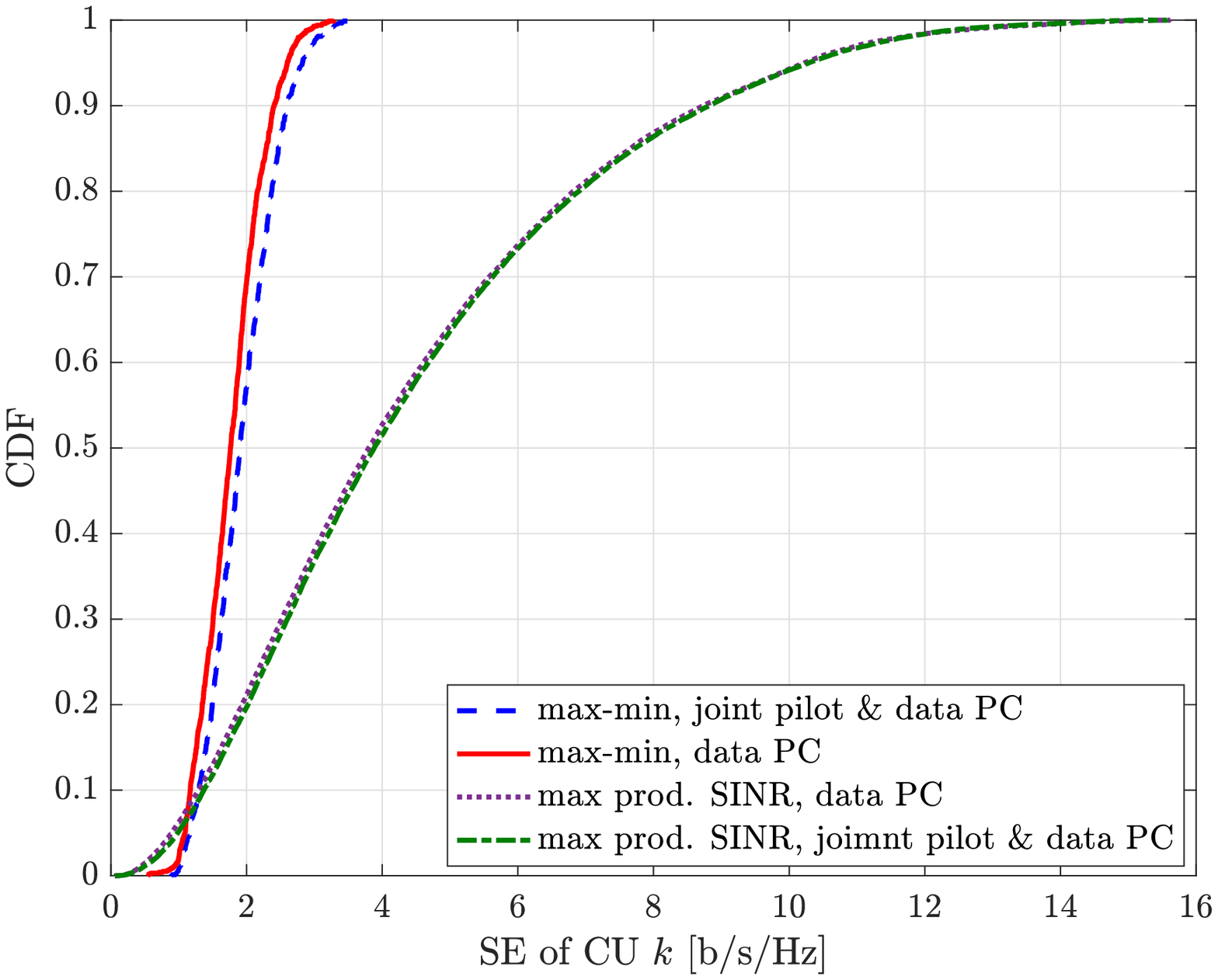} 
		\caption{ZF processing}  \vspace{-3mm}
		\label{fig:ZF data versus pilot}
	\end{subfigure} 
	\caption{SE of CU $k$ with data power control versus joint pilot and data power control.}
	\label{fig:data versus pilot}  
\end{figure}
\begin{figure}[hbpt!]
	\centering
	\begin{subfigure}[b]{0.8\columnwidth} \centering 
		\includegraphics[width=\columnwidth]{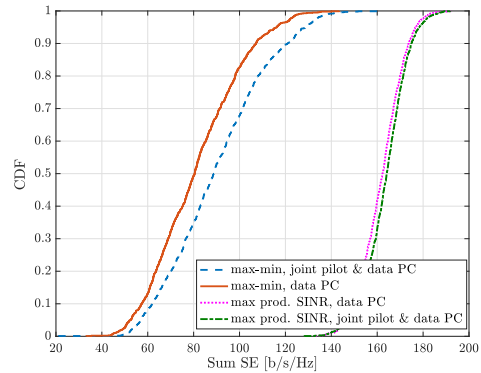} 
		\caption{MR processing}  \vspace{-3mm}
		\label{fig:Sum MRC data versus pilot}
	\end{subfigure} 
	\begin{subfigure}[b]{0.8\columnwidth} \centering
		\includegraphics[width=\columnwidth]{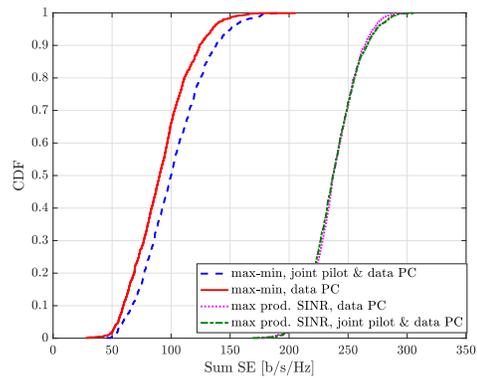} 
		\caption{ZF processing}  \vspace{-3mm}
		\label{fig:Sum ZF data versus pilot}
	\end{subfigure} 
	\caption{Sum SE of the network with data power control versus joint pilot and data power control.}
	\label{fig:Sum data versus pilot}  
\end{figure}

\section{Conclusion}
In this paper, we have presented a framework for pilot allocation for multi-cell massive MIMO systems with underlaid D2D communication. Different orthogonal sets of pilots were used for cellular and D2D communication in order to mitigate extra interference at the cellular BSs. First, we derived the SE expressions for D2D pairs and the SE expression of CUs in the system for both MR or ZF processing at the BSs. Then, we proposed different power control schemes for the considered system model and solved max-min fairness and max product SINR optimization problems for MR and ZF processing at the BSs. We considered both data power control and joint pilot and data power control for both optimization objectives. Most problems were solved to global optimality, but the non-convex joint pilot and data power control with ZF processing at the BSs required an iterative convex approximation algorithm which only provides a local optimal point. The simulation results show that all the algorithms can effectively limit the interference between CUs and D2D pairs to enable D2D underlaying with limited performance degradation for the CUs. The joint pilot and data power control enhances the SE of the network in comparison with data power control for max-min optimization objective for both ZF and MR processing. Interestingly, the optimal power control with max product SINR optimization is rather similar to using full power. In addition, the proposed joint pilot and data power control provide larger SE gains in the case of MR processing comparing with ZF processing.

\appendix

\subsection{Tightness of the approximate SE for D2D communication}\label{appendix:Approximate-SE}
Since the SE expression for D2D communication was not in closed form, thus an approximation was provided in \eqref{eq:app bound dd dedicated final} to enable tractable power control optimization. Fig.~\ref{fig:D2D bound} investigates the tightness of the approximation. We note that the approximation error is negligible for SEs less then $3\,$b/s/Hz, which is the range of main interest for max-min fairness optimization. For larger values of the SE, the approximation overestimates the SE by up to $0.5\,$b/s/Hz. Note that the D2D users are in relatively low SNR regime which is the reason for $0.5\,$b/s/Hz gap. When using the approximation for power control optimization, the D2D users will therefore get somewhat lower SE than what the output of the optimization problem predicts.
\begin{figure}[t!]
	\centering
	\includegraphics[width=0.8\columnwidth]{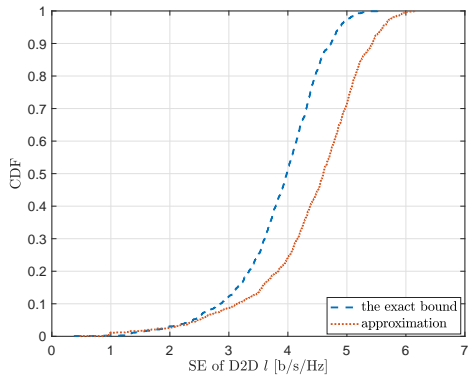}
	\caption{The approximate closed-form SE expression for D2D communication is tight for SEs less then $3\,$b/s/Hz.}
		\label{fig:D2D bound}
\end{figure}

\subsection{Basics of Geometric Programming}\label{appendix:Useful definitions}
In this appendix, we provide the basics of geometric programming. We start by defining two types of functions. A monomial function is defined as
\begin{equation}
\begin{aligned}
f({\mathbf{u}}) =& a \prod_{i} u^{b_i}_{i},\\
\end{aligned}
\end{equation}
where $a >0$, the exponents $b_i$ are real-valued, and $\mathbf{u}$ is a vector where each element $u_i$ is the real and positive \cite{boyd2007tutorial}. A posynomial function is formulated as the sum of one or multiple monomial functions and is written as
\begin{equation}
\begin{aligned}
f\left({\mathbf{u}}\right) = \sum_{k} a_k \prod_{i} u^{b_{i,k}}_{i},\\
\end{aligned}
\end{equation}
where~$a_k >0$ and the exponents $b_{i,k}$ are real-valued~\cite{boyd2007tutorial}. Adding or multiplying two posynomial functions produces a posynomial function as well. The result of dividing a posynomial function by a monomial function is also a posynomial. Dividing two posynomials result in a signomial function. This function has the same form as posynomial but the constants $a_k$ can be negative as well~\cite{boyd2007tutorial}.

A geometric program on standard form is expressed as
\begin{maxi!}[2]
	{\substack{\mathbf{u}\geq 0}}{f_0\left({\mathbf{u}}\right) \label{eq:gp-opti-objective}}
	{\label{eq:gp}}{}
	\addConstraint{ f_i\left({\mathbf{u}}\right)\leq 1,~~\forall i \label{eq:C1gp}}{}{}
	\addConstraint{ g_i\left({\mathbf{u}}\right) =  1,~~\forall i, \label{eq:C2gp}}{}{}
\end{maxi!}
where the objective function $f_0\left({\mathbf{u}}\right)$ and the constraint functions $f_i\left({\mathbf{u}}\right)$ are posynomials, while the constraint functions $g_i\left({\mathbf{u}}\right)$ are monomial~\cite{boyd2007tutorial}.

\subsection{Effect of D2D distance on spectral efficiency}\label{appendix:D2Ddisance}
The distance between the transmitter and receiver of each D2D pair is one of the parameters which affects the SE of CUs. By increasing this distance, the interference from D2D transmitters to the CUs increases. The reason behind this is that D2D transmitters have to transmit with higher powers to compensate for the larger pathloss. Fig.~\ref{fig:D2D_distance} shows the effect of increasing the D2D
distances on the SE of CU $k$ with MR processing at the BSs. It shows that by increasing this
distance from $10\,$m to $100\,$m, the SE of CU $k$ decreases significantly which verifies the importance of the distance between the transmitter and receiver of D2D pairs on the SE of CUs. $10\,$m is in range of the typical distances between transmitter and receiver of D2D pairs which has been considered in the literature \cite{shalmashi2016energy,liu2017pilot,Xu2017, He2017SE}. The results in Fig.~\ref{fig:D2D_distance} verify that this distance is reasonable as it has a negligible effect on the SE of CU $k$.
\begin{figure}[hbpt!]
	\centering
	\includegraphics[width=0.8\columnwidth]{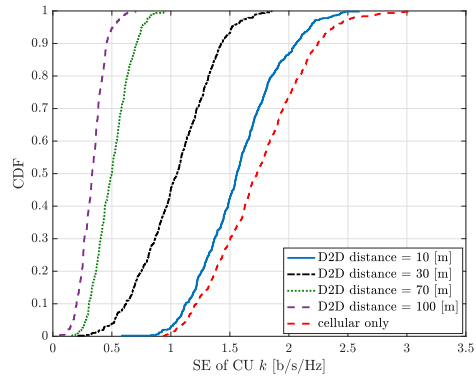}
	\caption{The effect of D2D distance on the SE of CU $k$.}
	\label{fig:D2D_distance}
\end{figure}

\subsection{Effect of ZF processing on data power coefficients}\label{appendix:power_coeff}
Fig.~\ref{fig:powerCoefficient} shows the power coefficients of the D2D transmitters and CUs for max SINR product and max-min fairness with ZF processing at the BSs. The results show that for the case of max product SINR the D2D transmitters and CUs can transmit with full power as ZF can efficiently mitigate the interference. This is the reason behind the similarity between the performance of this case and full power transmission.
\begin{figure}[hptb!]
	\centering
	\includegraphics[width=0.8\columnwidth]{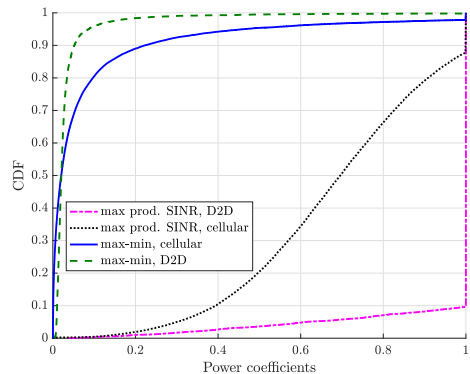}
	\caption{The power coefficients of the D2D transmitters and CUs for max SINR product and max-min fairness with ZF processing at the BSs.}
	\label{fig:powerCoefficient}
\end{figure}
\bibliographystyle{IEEEtran}

\bibliography{di}

\end{document}